\documentclass[12pt]{article}
\usepackage[margin=1in]{geometry}
\usepackage{amsmath,amsthm,amssymb,enumerate, graphicx,subcaption}
\usepackage{setspace}
\usepackage{thmtools}
\usepackage{thm-restate}
\usepackage{hhline}
\usepackage{cleveref,color}

\usepackage[mathscr]{euscript}
\usepackage[shortlabels]{enumitem}
\usepackage{prodint}

\declaretheorem[name=Theorem]{thm}
\declaretheorem[name=Proposition]{prop}
\declaretheorem[name=Lemma]{lemma}

\DeclareMathOperator*{\inprob}{\stackrel{P_0}{\longrightarrow}}

\DeclareMathOperator*{\indist}{\stackrel{d}{\longrightarrow}}
\newcommand{\bounded}{O_{P_0}}
\newcommand{\fasterthan}{o_{P_0}}
\newcommand{\fasterthandet}{o}
\newcommand{\boundeddet}{O}

\newcommand\independent{\protect\mathpalette{\protect\independenT}{\perp}}
\def\independenT#1#2{\mathrel{\rlap{$#1#2$}\mkern2mu{#1#2}}}

\newcommand{\s}[1]{\mathscr{#1}}
\renewcommand{\d}[1]{\mathbb{#1}}

\newcommand{\n}[1]{\mathrm{#1}}

\usepackage{natbib}

\allowdisplaybreaks

\title{\vspace{-2em}Nonparametric tests of the causal null \\with non-discrete exposures}
\author{Ted Westling\footnote{We are grateful for helpful feedback on this project from Dylan Small, Marco Carone, the University of Pennsylvania Causal Inference Working Group, and three anonymous referees.}\\
Department of Mathematics and Statistics\\
University of Massachusetts Amherst\\
twestling@math.umass.edu}
\date{}
\begin{document}

\maketitle

\begin{abstract}
In many scientific studies, it is of interest to determine whether an exposure has a causal effect on an outcome. In observational studies, this is a challenging task due to the presence of confounding variables that affect both the exposure and the outcome. Many methods have been developed to test for the presence of a causal effect when all such confounding variables are observed and when the exposure of interest is discrete. In this article, we propose a class of nonparametric tests of the null hypothesis that there is no average causal effect of an arbitrary univariate exposure on an outcome in the presence of observed confounding. Our tests apply to discrete, continuous, and mixed discrete-continuous exposures. We demonstrate that our proposed tests are doubly-robust consistent, that they have correct asymptotic type I error if both nuisance parameters involved in the problem are estimated at fast enough rates, and that they have power to detect local alternatives approaching the null at the rate $n^{-1/2}$. We study the performance of our tests in numerical studies, and use them to test for the presence of a causal effect of BMI on immune response in early-phase vaccine trials.
\end{abstract}

\doublespacing

\section{Introduction}

\subsection{Motivation and literature review}

One of the central goals of many scientific studies is to determine whether an exposure of interest has a causal effect on an outcome. In some cases, researchers are able to randomly assign units to exposure values. Classical statistical methods for assessing the association between two random variables can then be used to determine whether there is a causal effect because randomization ensures that there are no common causes of the exposure and the outcome. However, random assignment of units to exposures is not feasible in some settings, and even when it is feasible, preliminary evidence is often needed to justify a randomized study. In either case, it is often of interest to use data from an observational study, in which the exposure is not assigned by the researcher but instead varies according to some unknown mechanism, to assess the evidence of a causal effect. This is a more difficult task due to potential confounding of the exposure-outcome relationship.

Here, we are interested in assessing whether a non-randomized exposure that occurs at a single time point has a causal effect on a subsequent outcome when it can be assumed that there are no unobserved confounders. Many methods have been proposed to test the null hypothesis that there is no causal effect in this setting when the exposure is discrete. For instance, matching estimators \citep{rubin1973matching}, inverse probability weighted (IPW) estimators \citep{horvitz1952sampling}, and doubly-robust estimators including augmented IPW \citep{scharfstein1999adjusting, bang2005doubly} and targeted minimum loss-based estimators (TMLE) \citep{vanderlaan2011tmle} can all be used for this purpose. 

Much less work exists in the context of non-discrete exposures---that is, exposures that may take an uncountably infinite number of values. In practice, researchers often discretize such an exposure in order to return to the discrete exposure setting. This simple approach has several drawbacks. First, since the results often vary with the choice of discretization, it may be tempting for researchers to choose a discretization based on the results. However, this can inflate type I error rate. Second, tests based on a discretized exposure typically have less power than tests based on the original exposure because discretizing discards possibly relevant information (see, e.g.\ \citealp{cox1957grouping, cohen1983cost, fedorov2009consequence}). Finally, causal estimates based on a discretized exposure have a more complicated interpretation than those based on the original exposure \citep{young2019representative}.

In the context of causal inference with a continuous exposure, one common estimand is the \emph{causal dose-response curve}, which is defined for each value of the exposure as the average outcome were all units assigned to that exposure value. We say there is no average causal effect if the dose-response curve is flat; i.e.\ the average outcome does not depend on the assigned value of the exposure. One approach to estimating the dose-response curve is to assume the regression of the outcome on the exposure and confounders follows a linear model. If the model is correctly specified, the coefficient corresponding to the exposure is the slope of the dose-response curve, and the null hypothesis that the dose-response curve is flat can be assessed by testing whether this coefficient is zero. This approach can be generalized to other regression models, which can be marginalized using the G-formula to obtain the dose-response function \citep{robins1986,robins2000msm, zhang2016quantitative}. However, if the regression model is mis-specified, then the resulting estimator of the causal dose-response curve is typically inconsistent, and the resulting test will not necessarily be calibrated or consistent. Inverse probability weighting may also be used to estimate the dose-response curve \citep{imai2004gps, hirano2005gps}, but mis-specification of the propensity score model may again lead to unreliable inference.

Nonparametric methods make fewer assumptions about the data-generating mechanism than methods based on parametric models, and are therefore often more robust. In the context of nonparametric estimation of a causal dose-response curve, \cite{neugebauer2007jspi} considered  inference for the projection onto a parametric working model; \cite{rubin2006msm} and \cite{diaz2011super} discussed the use of data-adaptive algorithms; \cite{kennedy2016continuous} and \cite{van2018non} proposed estimators based on kernel smoothing; and \cite{westling2020isotonic} proposed an isotonic estimator. However, none of these works considered tests of the null hypothesis that the dose-response curve is flat. We also note that tests of Granger causality \citep{granger1980testing} have been developed in the context of time series (e.g.\ \citealp{granger1995causality, boudjellaba1992arma,nishiyama2011testing}). This is distinct from our goal, which is to test whether an exposure at a single time point has a causal effect on the average of a subsequent outcome.

\subsection{Contribution and organization of the article}

In this article, we focus on the problem of testing the null hypothesis that there is no average causal effect of a possibly non-discrete exposure on an outcome against the complementary alternative. To the best of our knowledge, no nonparametric test has yet been developed for this purpose. Specifically, we (1) propose a test based on a cross-fitted nonparametric estimator of an integral of the causal dose-response curve; (2) provide conditions under which our test has desirable large-sample properties, including (i) consistency under any alternative as long as either of two nuisance functions involved in the problem is estimated consistently (known as \emph{doubly-robust} consistency), (ii) asymptotically correct type I error rate, and (iii) non-zero power under local alternatives approaching the null at the rate $n^{-1/2}$; and (3) illustrate the practical performance of the tests through numerical studies and an analysis of the causal effect of BMI on immune response in early-phase vaccine trials.

Notably, the conditions we establish for consistency and validity of our test do not restrict the form of the marginal distribution of the exposure. Therefore, our test applies equally to discrete, continuous, and mixed discrete-continuous exposures. In the second set of numerical studies, we demonstrate that even in the context of discrete exposures, existing tests do not control type I error when the number of discrete components is moderate or large relative to sample size, while the tests proposed here are valid in all such circumstances.

The remainder of the article is organized as follows. In Section~\ref{sec:methods}, we define our proposed procedure. In Section~\ref{sec:asymptotic}, we discuss the large-sample properties of our procedure. In Section~\ref{sec:simulation}, we illustrate the behavior of our method using numerical studies. In Section~\ref{sec:bmi}, we use our procedure to analyze the causal effect of BMI on immune response. Section~\ref{sec:discussion} presents a brief discussion. Proofs of all theorems are provided in Supplementary Material. An \texttt{R} \citep{Rlang} package implementing all the methods developed in this paper is available at \texttt{https://github.com/tedwestling/ctsCausal}.

\section{Proposed methodology}\label{sec:methods}

\subsection{Notation and null hypothesis of interest}

We denote by $A \in \s{A}$ the real-valued exposure of interest, whose support we denote by $\s{A}_0 \subseteq \d{R}$. Adopting the Neyman-Rubin potential outcomes framework, for each $a \in \s{A}_0$, we denote by $Y(a) \in \s{Y} \subseteq\d{R}$ a unit's potential outcome under an intervention setting exposure to $A =a$. The causal parameter $m(a) := E\left[Y(a)\right]$ represents the average outcome under assignment of the entire population to exposure level $A=a$. The resulting curve $m:\s{A}\rightarrow \d{R}$ is known as the \emph{causal dose-response curve}. 

We are interested in testing the null hypothesis that $m(a) = \gamma_0$ for all $a \in \s{A}_0$ and some $\gamma_0 \in \d{R}$; that is, the dose-response curve is flat. This null hypothesis holds if and only if the average value of the potential outcome is unaffected by the value of the exposure to which units are assigned. However, $Y(a)$ is not typically observed for all units in the population, but instead the outcome $Y := Y(A)$ corresponding to the exposure value actually received is observed. Thus, $m$ is not a mapping of the joint distribution of the pair $(A, Y)$, so the null hypothesis that $m$ is flat cannot be tested using this data. The first step in developing a testing procedure is to translate the causal problem into a problem that is testable with the observed data, which is called \emph{identification} in the causal inference literature.

We first assume that (i) each unit's potential outcomes are independent of all other units' exposures 
and (ii) the observed outcome $Y$ almost surely equals $Y(A)$.  
If (i)--(ii) hold and in addition (iii) $Y(a) \independent A$ for all $a \in \s{A}_0$, then $m(a) = E[Y \mid A = a]$ for all $a \in \s{A}_0$. In this case, $m$ is identified with a univariate regression function, so the null hypothesis that $m$ is flat on $\s{A}_0$ can be tested using existing work from the nonparametric regression literature (see, e.g.\ \citealp{eubank1990testing, andrews1997kolmogorov, horowitz2001adaptive}, among many others). However, assumption (iii) typically only holds in experiments in which $A$ is randomly assigned. In observational studies, there are typically \emph{confounding variables} that impact both $A$ and $Y(a)$, thus invalidating (iii). In these settings, tests based on nonparametric regression may not have valid type I error rate, even asymptotically.

We suppose that a collection $W \in \s{W} \subseteq \d{R}^p$ of possible confounders is recorded.  If (i)--(ii) hold and in addition (iv) $Y(a) \independent A \mid W$ for all $a \in \s{A}_0$, known as \emph{no unmeasured confounding}, and (v) all $a \in \s{A}_0$ are in the support of the conditional distribution of $A$ given $W = w$ for almost every $w$, known as \emph{positivity}, then $m(a)= \theta_0(a) := E[ E(Y \mid A=a, W)]$, which is known as the G-computed regression function \citep{robins1986, gill2001}. Therefore, under (i)--(ii) and (iv)--(v), $m$ is flat on $\s{A}_0$ if and only if $\theta_0$ is flat on $\s{A}_0$. 

Here, we assume that we observe independent and identically distributed random vectors $(Y_1, A_1, W_1), \dotsc, (Y_n, A_n,W_n)$ from a distribution $P_0$ contained in the nonparametric model $\s{M}_{NP}$ consisting of all distributions on $\s{Y} \times \s{A} \times \s{W}$.  For a distribution $P \in \s{M}_{NP}$, we denote by $F_P$ the marginal distribution of $A$ under $P$, $\s{A}_P$ the support of $F_P$, $\mu_P(a,w) := E_P[Y \mid A =a, W=w]$ the outcome regression function, and $Q_P$ the marginal distribution of $W$. Throughout, we use the subscript $0$ to refer to evaluation at or under $P_0$; for example, we denote by $F_0$ the marginal distribution function of $A$ under $P_0$ and by $\s{A}_0$ the support of $F_0$.  Denoting by $C_b(S)$ the class of continuous and bounded functions on a subset $S$ of $\d{R}$, we assume that $P_0$ is contained in the statistical model $\s{M} := \{ P \in \s{M}_{NP}: \theta_P \in C_b(\s{A}_P)\}$.

In this article, we focus on testing the null hypothesis $H_0 : P_0 \in \s{M}_0 \subset \s{M}$ for $\s{M}_0 := \left\{ P \in \s{M} : \theta_P(a) = \theta_P(a') \text{ for all } a, a' \in \s{A}_P \right\}$ versus the complementary alternative  $H_A : P_0 \in \s{M}_A :=  \s{M} \backslash \s{M}_0$. 
This null hypothesis holds if and only if $\theta_0(a) = \gamma_0$ for all $a \in \s{A}_0$, where $\gamma_0 :=  \int \theta_0(a) \,F_0(da) = \iint \mu_0(a, w) \, Q_0(dw) \,F_0(da)$. As noted above, conditions (i)--(ii) and (iv)--(v) imply that $H_0$ holds if and only if the exposure has no causal effect on the average outcome in the sense that setting $A$ to $a$ for all units in the population yields the same average outcome for all $a \in \s{A}_0$. On the other hand, $H_A$ holds if and only if at least two exposures in $\s{A}_0$ yield different average outcomes. Our null hypothesis is stated in terms of the possibly unknown $\s{A}_0$ because condition (v) does not hold for $a \notin \s{A}_0$, and in fact $m(a)$ is not identified in the observed data for such $a$ without further assumptions.

A special case of our null hypothesis is that $\mu_0(a, w) = \mu_0(a', w)$ for all $a, a' \in \s{A}_0$ and almost all $w$. In this case, $\gamma_0 = E_0[Y]$. Under conditions (i)--(ii) and (iv)--(v), this happens if and only if $E_0[Y(a) \mid W = w] =E_0[Y(a') \mid W = w]$ for almost all $w$---i.e. there is no effect of the exposure on the average potential outcome for any strata of $W$ in the population. We shall refer to this case as the \emph{strong} null hypothesis. We emphasize that our null hypothesis $H_0$ can hold even if the strong null does not, since interactions between the exposure and covariates may average out to yield a flat G-computed regression curve. We note that \cite{luedtke2019omnibus} recently proposed a general procedure for testing null hypotheses of the form $R_0(O) \stackrel{d}{=} S_0(O)$, where the generic observation $O$ follows distribution $P_0$, and the functions $R_0$ and $S_0$ may depend on $P_0$. \cite{luedtke2019omnibus} demonstrated that their procedure can be used to consistently test the strong null hypothesis stated above, albeit with type I error rate tending to zero.  Our null hypothesis may also be stated in their general form with $R_0(Y,A,W) := \theta_0(A)$ and $S_0(Y, A, W) := \gamma_0$. However, their results do not apply to our null hypothesis because their Condition~3 does not hold for $R_0 = \theta_0$.

For a measure $\lambda$, we define $\| h \|_{\lambda, p} := \left[ \int | h(x) |^p \, d\lambda(x) \right]^{1/p}$ for $p \in [1,\infty)$, and $\| h \|_{\lambda, \infty} := \sup_{x \in \n{supp}(\lambda)} |h(x)|$. For a probability measure $P$ and $P$-integrable function $h$, we define $Ph := \int h \, dP$. We define $\d{P}_n$ as the empirical distribution function of the observed data.

\subsection{Testing in terms of the primitive function}

Our procedure will be based on estimating a primitive parameter (i.e.\ an integral) of $\theta_0$.  A useful analogy to consider is testing that a density function $h$ is flat on $[0,1]$. A density function, like $\theta_0$, is a challenging parameter to estimate in a nonparametric model, and the fastest attainable rate is slower than $n^{-1/2}$. However, $h$ is flat if and only if the corresponding cumulative distribution function $H$, which is the primitive of $h$, is the identity function on $[0,1]$, which is further equivalent to $H(x) - x = 0$ for all $x \in [0,1]$. In our setting, we define the primitive function as  $\Gamma_0(a) := \int_{-\infty}^a \theta_0(u) \, dF_0(u)$. We integrate against the marginal distribution $F_0$ of $A$ because $\theta_0$ is only estimable on the support of $F_0$. We then note that $\theta_0(a) = \gamma_0$  for all $a \in \s{A}_0$ if and only if $\Gamma_0(a) = \int_{-\infty}^a \gamma_0 \, dF_0(u) = \gamma_0F(a)$. Therefore, our null hypothesis is equivalent to $\Omega_0(a) := \Gamma_0(a) - \gamma_0 F_0(a)$ equals $0$ for all $a \in \s{A}_0$. This is stated formally in the following result.
\begin{prop}\label{prop:equivalence}
If $\theta_0$ is continuous on $\s{A}_0$, then the following are equivalent: (1) $P_0 \in \s{M}_0$; (2) $\theta_0(a) = \gamma_0$ for all $a \in \s{A}_0$; (3) $\Omega_0(a) = 0$ for all $a \in \d{R}$; (4) $\|\Omega_0\|_{F_0, p} = 0$ for all $p \geq 1$.
\end{prop}
\begin{figure}[t!]
\centering
\includegraphics[width=5in]{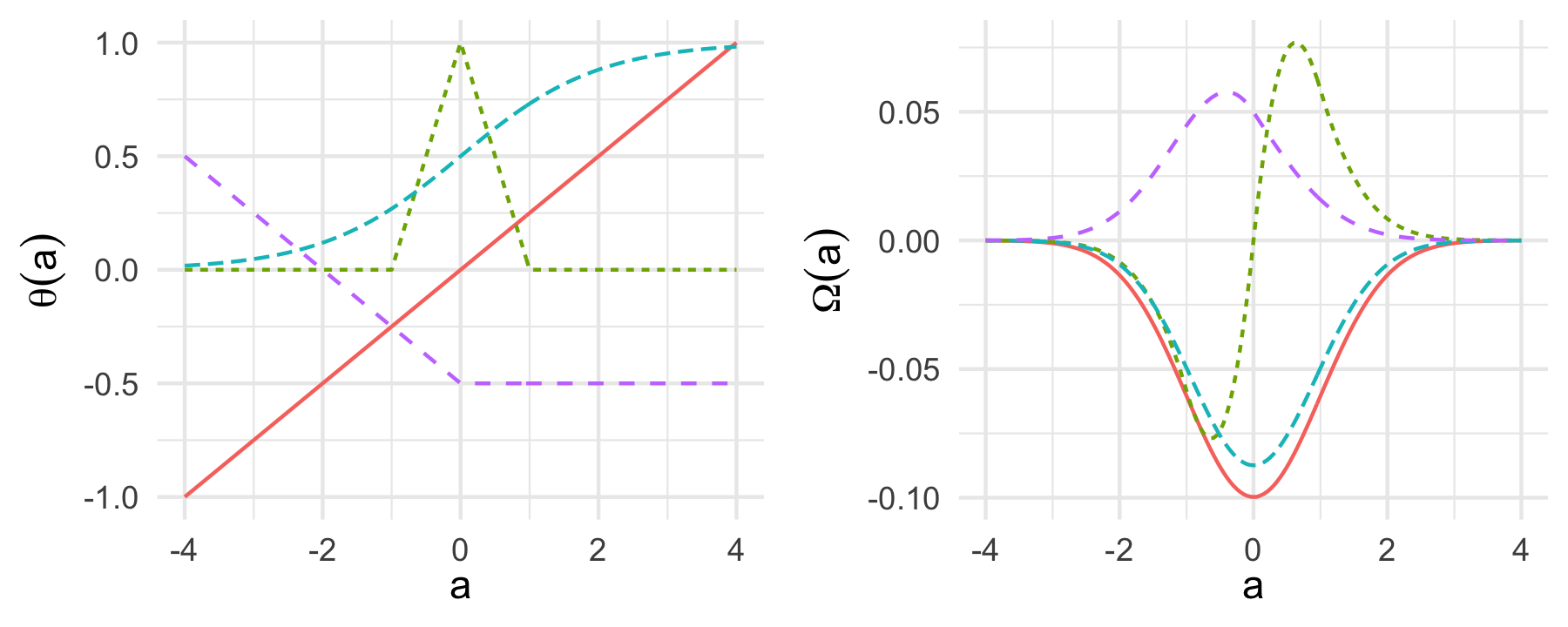}
\caption{Four hypothetical dose-response curves $\theta(a)$ (left) and their corresponding primitives $\Omega(a)$ (right). The marginal distribution of the exposure is the standard normal.}
\label{fig:theta_omega}
\end{figure}
For illustrative purposes, Figure~\ref{fig:theta_omega} displays four hypothetical $\theta_0$'s and their corresponding $\Omega_0$'s. We note that $\Omega_0(a) = 0$ for all $a \in \d{R}$ if and only if $\Omega_0(a) = 0$ for all $a \in \s{A}_0$. This, combined with Proposition~\ref{prop:equivalence}, indicates that in the model $\s{M}$, testing $H_0$ is equivalent to testing the null hypothesis $\|\Omega_0\|_{F_0, p} = 0$ versus the alternative $\|\Omega_0\|_{F_0, p} > 0$. In the case of testing uniformity of a density, this is analogous to testing that $\|H - \n{Id}\|_{[0,1],p} = 0$ for $\n{Id}(x) = x$ the identity map, which is exactly what the Kolmogorov-Smirnov  (in which $p = \infty$) and Anderson-Darling   (in which $p = 2$) tests do. Furthermore, representing the problem as a test that $\|H - \n{Id}\|_{[0,1],p} = 0$ is useful because $H$ is staightforward to estimate at the $n^{-1/2}$ rate, unlike a density function. In our case, $\Omega_0$ is pathwise differentiable in the nonparametric model with an estimable influence function under standard conditions, unlike $\theta_0$.  Specifically, defining $g_0(a, w) := G_0(da, w) / F_0(da)$, where  $G_0(a , w) := P_0(A \leq a \mid W = w)$ is the conditional distribution of $A$ given $W =w$ evaluated at $a$, we have the following result. 
\begin{prop}\label{prop:eif}
For each $a_0 \in \d{R}$, $\Omega_0(a_0)$ is pathwise differentiable relative to the subset of $\s{M}$ in which $g_P$ is almost surely bounded away from zero and $E_P\left[Y^2 \right] < \infty$, and its efficient influence function is 
\begin{align*}
D_{a_0,0}^*(y, a, w) &:= \left[I_{(-\infty, a_0]}(a) - F_0(a_0)\right] \left[\frac{ y - \mu_0(a,w)}{g_0(a, w)} +  \theta_{0}(a) - \gamma_0\right] \\
&\qquad +\int \left[I_{(-\infty, a_0]}(u) - F_0(a_0)\right]   \mu_0(u, w) \, F_0(du)  - 2\Omega_0(a_0)\ .
\end{align*}
\end{prop}
We note that $g_0(a, w) = P_0(A = a \mid W =w) / P_0(A = a)$ for $a$ such that $P_0(A = a) > 0$, and $g_0(a, w) = \left[\frac{d}{da} P_0( A \leq a \mid W =w)\right] /\left[ \frac{d}{da} F_0(a)\right]$ for $a$ where $F_0$ is absolutely continuous. Hence, $g_0$ generalizes both the standardized generalized propensity score (GPS) and the standardized propensity score.
 
For any $p \in [1,\infty]$, Proposition~\ref{prop:eif} allows us to test $H_0$ in the following manner: (1) construct a uniformly asymptotically linear estimator $\Omega_n^\circ$ of $\Omega_0$ for which $\{n^{1/2}[ \Omega_n^\circ(a) - \Omega_0(a)] : a \in \s{A}_0\}$ converges weakly to a Gaussian limit process, (2) use the estimated influence function of $\Omega_n^\circ(a)$ to construct an estimator $T_{n,p,\alpha}$ of the $1-\alpha$ quantile of $n^{1/2}\|\Omega_n^\circ - \Omega_0\|_{F_0,p}$, and (3) reject $H_0$ at level $\alpha$ if $n^{1/2}\|\Omega_n^\circ \|_{F_n,p} > T_{n,p, \alpha}$. In the remainder of this section, we provide details for accomplishing each of these three steps.

\subsection{Estimating the primitive function}\label{sec:primitive_est}

The first step in our testing procedure is to construct an asymptotically linear estimator of $\Omega_0(a)$ for each fixed $a$. The four nuisance parameters present in the definition of $\Omega_0(a_0)$ and its nonparametric efficient influence function $D_{a_0,0}^*$ are the outcome regression $\mu_0$, the standardized propensity $g_0$, and the marginal distributions  $F_0$ and $Q_0$ corresponding to $A$ and $W$, respectively. Given estimators $\mu_n$ and $g_n$ of $\mu_0$ and $g_0$, respectively, we can construct an estimator $D_{a_0,n}^*$ of the influence function by plugging in $\mu_n$ for $\mu_0$, $g_n$ for $g_0$, and the empirical marginal distributions $F_n$ and $Q_n$ for $F_0$ and $Q_0$. A one-step estimator of $\Omega_0(a_0)$ is then given by $\Omega_n(a_0) := \Omega_{\mu_n, F_n, Q_n}(a_0)+ \d{P}_n D_{a_0,n}^*$, where $\Omega_{\mu_n, F_n, Q_n}(a_0) := \iint \left[ I_{(-\infty, a_0]}(a) - F_n(a_0)\right] \mu_n(a, w) \, dF_n(a) \, dQ_n(w)$ is the plug-in estimator of $\Omega_0$. In expanding $\Omega_n(a_0)$, some terms cancel and we are left with
\begin{equation}
\Omega_n(a_0) = \frac{1}{n}\sum_{i=1}^n  \left[ I_{(-\infty, a_0]}(A_i) - F_n(a_0) \right] \left[\frac{ Y_i - \mu_n(A_i,W_i)}{g_n(A_i, W_i)} + \int \mu_n(A_i, w) \, dQ_n(w) \right]  \ .\label{eq:one_step_est}
\end{equation}
If we were to base our test on $\Omega_n$, then, as we will see in Section~\ref{sec:asymptotic}, the large-sample properties of our test would depend on consistency of $\Omega_n$ and on weak convergence of $\{ n^{1/2}[\Omega_n(a) - \Omega_0(a)] : a\in\s{A}_0\}$ as a process. Such statistical properties of asymptotically linear estimators of pathwise differentiable parameters depend on estimators of nuisance parameters in two important ways. First, negligibility of a so-called \emph{second-order} remainder term requires negligibility of $(\mu_n - \mu_0)(g_n - g_0)$ in an appropriate sense. Second, negligibility of \emph{empirical process} remainder terms can be guaranteed if the nuisance estimators fall in sufficiently small function classes. In observational studies, researchers can rarely specify a priori correct parametric models for $\mu_0$ or $g_0$, which motivates the use of data-adaptive (e.g.\ machine learning) estimation of these functions in order to achieve negligibility of the second-order remainder. However, data-adaptive estimators  typically constitute large function classes. Hence, finding estimators that simultaneously satisfy these two requirements can require a delicate balance. Cross-fitting has been found to resolve this challenge by removing the need for nuisance estimators to fall in small function classes \citep{zheng2011cvtmle, belloni2018uniform, kennedy2019incremental}. Instead of basing our test on $\Omega_n$, we will therefore base our test on a cross-fitted version of $\Omega_n$, which we now define.

For a deterministic integer $V \in \{2, 3, \dotsc, \lfloor n/2\rfloor \}$, we randomly partition the indices $\{1, \dotsc, n\}$ into $V$ disjoint sets $\s{V}_{n,1}, \dotsc, \s{V}_{n,V}$ with cardinalities $N_1, \dotsc, N_V$. We require that these sets be as close to equal sizes as possible, so that $|N_v - n/V| \leq 1$ for each $v$, and that the number of folds $V$ be bounded as $n$ grows. For each $v \in \{1, \dotsc, V\}$, we define $\s{T}_{n,v} := \{ O_i : i \notin \s{V}_{n,v}\}$ as the \emph{training set} for fold $v$, and we define $\mu_{n,v}$ and $g_{n,v}$ as nuisance estimators that are estimated using only the observations from $\s{T}_{n,v}$. Similarly, we define $F_{n,v}$ and $Q_{n,v}$ as the marginal empirical distributions of $A$ and $W$, respectively, corresponding to the observations in $\s{T}_{n,v}$. We then define the cross-fitted estimator $\Omega_n^\circ$ of $\Omega_0$ as 
\begin{align}\Omega_{n}^\circ(a_0)  &:=\frac{1}{V} \sum_{v=1}^V \left\{ \frac{1}{N_v}\sum_{i \in \s{V}_{n,v}} \left[ I_{(-\infty, a_0]}(A_i) - F_{n,v}(a_0)\right] \frac{ Y_i - \mu_{n,v}(A_i,W_i)}{g_{n,v}(A_i, W_i)} \right.\nonumber\\
&\qquad\qquad\qquad \left. + \frac{1}{N_v^2}\sum_{i, j \in \s{V}_{n,v}} \left[ I_{(-\infty, a_0]}(A_i) - F_{n,v}(a_0)\right] \mu_{n,v}(A_i, W_j) \right\}\ . \label{eq:one_step_cv_est} \end{align}

In the next section, we indicate properties of the estimators $\mu_{n,v}$ and $g_{n,v}$ that imply certain large-sample properties of $\Omega_n^\circ$, which in turn imply properties of our testing procedure. In particular, we provide conditions under which $\Omega_n^\circ(a)$ is uniformly asymptotically linear with influence function $D_{a,0}^*$, meaning that
\begin{equation} \Omega_n^\circ(a) - \Omega_0(a) = \d{P}_n D_{a,0}^* + R_n(a) \ ,\label{eq:unif_asy_lin} \end{equation}
where $\sup_{a\in\s{A}_0} |R_{n}(a)| = \fasterthan(n^{-1/2})$. If \eqref{eq:unif_asy_lin} holds and in addition the one-dimensional class of functions $\{D_{a, 0}^* : a \in \s{A}_0\}$ is $P_0$-Donsker, then $\{ n^{1/2}[\Omega_n^\circ(a) - \Omega_0(a)] : a \in \s{A}_0 \}$ converges weakly in the space $\ell^{\infty}(\s{A}_0)$ of bounded real-valued functions on $\s{A}_0$ to a mean-zero Gaussian process $Z_0$ with covariance function $\Sigma_0(s,t) := P_0 [ D_{s, 0}^* D_{t,0}^*]$. Since the $L_p(F_0)$-norm is a continuous functional on  $\ell^{\infty}(\s{A}_0)$ for any $p \in [1, \infty]$, by the continuous mapping theorem we will then have $n^{1/2} \| \Omega_n^\circ - \Omega_0\|_{F_0,p} \indist \| Z_0 \|_{F_0, p}$. Given an estimator $D_{a,n,v}^*$ of $D_{a,0}^*$ for each $v$, we can approximate the distribution $\| Z_0 \|_{F_0, p}$ by simulating sample paths of a mean-zero Gaussian process $Z_n$ with covariance function $\Sigma_n(s,t) := \frac{1}{V} \sum_{v=1}^V \d{P}_{n,v} D_{s,n,v}^* D_{t,n,v}^*$, and computing the $L_p(F_n)$-norm of these sample paths, where $\d{P}_{n,v}$ is the empirical distribution for the validation fold $\s{V}_{n,v}$.  Putting it all together, our fully specified procedure for testing the null hypotheses $H_0$ is as follows:
\begin{description}[style=multiline,leftmargin=1.8cm]
\item[Step 1:] Split the sample into $V$ sets $\s{V}_{n,1}, \dotsc, \s{V}_{n,V}$ of approximately equal size.
\item[Step 2:] For each $v \in \{1, \dotsc, V\}$, construct estimates $\mu_{n,v}$ and $g_{n,v}$ of the nuisance functions $\mu_0$ and $g_0$ based on the training set $\s{T}_{n,v}$ for fold $v$.
\item[Step 3:] For each $a$ in the observed values of the exposure $\s{A}_n := \{A_1, \dotsc, A_n\}$, use $\mu_{n,v}$ and $g_{n,v}$ to construct $\Omega_n^\circ(a)$ as defined in \eqref{eq:one_step_cv_est}.
\item[Step 4:] Let $T_{n, \alpha, p}$ be the $1-\alpha$ quantile of $ \left( \frac{1}{n} \sum_{i=1}^n |Z_{n}(A_i)|^p\right)^{1/p}$ for $p < \infty$ or $\max_{a \in \s{A}_n} |Z_n(A_i)|$ for $p = \infty$, where, conditional on $O_1, \dotsc, O_n$, $(Z_n(A_1), \dotsc, Z_n(A_n))$ is distributed according to a mean-zero multivariate normal distribution with covariances given by $\Sigma_n(A_i,A_j) := E[ Z_{n}(A_i) Z_n(A_j) \mid O_1, \dotsc, O_n] = \frac{1}{V} \sum_{v=1}^V \d{P}_{n,v} D_{A_i,n,v}^*D_{A_j,n,v}^*$ for
\begin{align*}
  D_{a_0,n,v}^*(y, a, w)&=\left[I_{(-\infty, a_0]}(a) - F_{n,v}(a_0)\right] \left[\frac{ y - \mu_{n,v}(a,w)}{g_{n,v}(a, w)} +  \theta_{n,v}(a) -  \gamma_{n,v}\right] \\
  &\qquad + \int \left[I_{(-\infty, a_0]}(u) - F_{n,v}(a_0)\right]   \mu_{n,v}(u, w) \, F_{n,v}(du)  - 2\Omega_{\mu_{n,v}, F_{n,v}, Q_{n,v}}(a_0) \ , 
  \end{align*}
where $\theta_{n,v}(a) := \int \mu_{n,v}(a, w) \, dQ_{n,v}(w)$ and $\gamma_{n,v} := \iint \mu_{n,v}(a, w) \, dF_{n,v}(a) \, dQ_{n,v}(w)$.
\item[Step 5:] Reject $H_0$ at level $\alpha$ if $n^{1/2}\|\Omega_n^\circ \|_{F_n, p} > T_{n, \alpha, p}$.
\end{description}

In practice, we recommend using $p=\infty$ for several reasons. First, as illustrated in the numerical studies, tests based on $p=\infty$ offer better finite-sample power for detecting non-linear alternatives than other $p$, at no cost to test size. Relatedly, we expect tests based on $p= \infty$ to be more sensitive to deviations from the null that are concentrated on a narrow region of the support of $F_0$. Second, unlike $\|\Omega_0\|_{F_0, p}$ for $p < \infty$, $\|\Omega_0\|_{F_0, \infty}$ only depends on $F_0$ through its support, which makes it a more interpretable and generalizable parameter.

%

\section{Asymptotic properties of the proposed procedure}\label{sec:asymptotic}

\subsection{Doubly-robust consistency}

In this section, we derive sufficient conditions for three large-sample properties of our proposed test: consistency under fixed alternatives, asymptotically correct type I error rate, and positive asymptotic power under local alternatives. Each of these three properties is established by first proving an accompanying result for the estimator $\Omega_n^\circ$ upon which the test is based. We start by showing that the proposed test is doubly-robust consistent, meaning that it rejects any alternative hypothesis with probability tending to one as long as either of the two nuisance parameters involved in the problem is estimated consistently. We first introduce several conditions upon which our results rely.

\begin{description}[style=multiline,leftmargin=1.1cm]
\item[(A1)]  There exist $K_0, K_1, K_2 \in (0, \infty)$ such that, almost surely as $n \to \infty$ and for all $v$, $\mu_{n,v}$ and $\mu_0$ are contained in a class of functions $\s{F}_0$ and $g_{n,v}$ and $g_0$ are contained in a class of functions $\s{F}_{1}$, where $|\mu| \leq K_{0}$ for all $\mu \in \s{F}_{0}$ and $K_{1} \leq g \leq K_{2}$ for all $g \in \s{F}_{1}$. Additionally, $E_0[Y^2] < \infty$.
\item[(A2)] There exist $\mu_{\infty} \in \s{F}_0$ and $g_{\infty} \in \s{F}_1$ such that $\max_v P_0 (\mu_{n,v} - \mu_{\infty})^2 \inprob 0$ and $\max_v P_0(g_{n,v} - g_{\infty})^2  \inprob 0$.
\item[(A3)] There exist subsets $\s{S}_1, \s{S}_2$ and $\s{S}_3$ of $\s{A}_0 \times\s{W}$ such that $P_0(\s{S}_1 \cup \s{S}_2 \cup \s{S}_3) = 1$ and:
\begin{enumerate}[(a)]
\item $\mu_{\infty}(a,w) = \mu_0(a,w)$ for all $(a,w) \in \s{S}_1$;
\item $g_{\infty}(a,w) = g_0(a,w)$ for all $(a,w) \in \s{S}_2$;
\item $\mu_{\infty}(a,w) = \mu_0(a,w)$ and $g_{\infty}(a,w) = g_0(a,w)$ for all $(a,w) \in \s{S}_3$.
\end{enumerate}
\end{description}
Condition (A1) requires that the true nuisance functions as well as their estimators satisfy certain boundedness constraints. The requirement that $g_0 \geq K_1 > 0$ is a type of \emph{overlap} or \emph{positivity} condition. For mass points $a \in \s{A}_0$, this is equivalent to requiring that $P_0(A = a_0 \mid W = w) / P_0(A = a_0) \geq K_1 > 0$ for almost every $w$ and every such $a$. If there are a finite number of mass points, then this condition is equivalent to the standard overlap condition used for $n^{-1/2}$-rate estimation with a discrete exposure. However, if there are an infinite number of mass points, then the condition is weaker than the standard overlap condition. Similarly, for points $a \in \s{A}_0$ where $F_0$ is absolutely continuous, (A1) is related to but strictly weaker than the standard overlap condition for estimation of a dose-response curve with a continuous exposure, which requires that the conditional density $p_0(a \mid w)$ be bounded away from zero for almost all $w$ (e.g.\ condition (e) of Theorem 2 of \citealp{kennedy2016continuous}). Instead, (A1) requires that $p_0(a \mid w) / f_0(a)$ be bounded away from zero, which may hold even when $p_0(a \mid w)$ is arbitrarily close to zero. For example, if $A$ and $W$ are independent, so that $p_0(a \mid w) = f_0(a)$, then (A1) is automatically satisfied, whereas the standard overlap condition would not necessarily be.

Condition (A2) requires that the nuisance estimators converge to some limits $\mu_\infty$ and $g_\infty$. Condition (A3) is known as a double-robustness condition, since it is satisfied if either $\mu_\infty = \mu_0$ almost surely or $g_\infty = g_0$ almost surely. Double-robustness has been studied for over two decades, and is now commonplace in causal inference \citep{robins1994estimation, rotnitzky1998semiparametric, scharfstein1999adjusting, van2003unified, neugebauer2005prefer, bang2005doubly}. However, condition (A3) is slightly more general than standard double-robustness, since it is satisfied if either $\mu_\infty(a,w) = \mu_0(a,w)$ or $g_\infty(a,w) = g_0(a,w)$ for almost all $(a,w)$, which can happen even if neither $\mu_\infty = \mu_0$ nor $g_\infty = g_0$ almost surely.

Under these conditions, we have the following result concerning consistency of $\Omega_n^\circ$.
\begin{thm}[Doubly-robust consistency of $\Omega_n^\circ$]\label{thm:dr_cons_omega}
If (A1)--(A3) hold, then $\sup_{a \in \d{R}}|\Omega_n^\circ(a) - \Omega_0(a)| \inprob 0.$
\end{thm}
It follows immediately from Theorem~\ref{thm:dr_cons_omega} that if (A1)--(A3) hold, then $\|\Omega_n^\circ\|_{F_0, p} \inprob \|\Omega_0\|_{F_0,p}$ for any $p \in [1, \infty]$, so that $P_0\left(\|\Omega_n^\circ\|_{F_0, p} > t_n\right) \longrightarrow 1$ for any $t_n \inprob 0$ and $P_0 \in \s{M}$ such that $H_A$ holds. In order to fully establish consistency of the proposed test, we need to justify using $\|\cdot\|_{F_n,p}$ instead of $\|\cdot\|_{F_0,p}$, and in addition we need to show that $T_{n,\alpha, p} / n^{1/2} \inprob 0$. We do so in the next result, and conclude that the proposed test is doubly-robust consistent.
\begin{thm}[Doubly-robust consistency of proposed test]\label{thm:dr_cons_test}
If conditions (A1)--(A3) hold, then $P_0\left(n^{1/2}\|\Omega_n^\circ \|_{F_n, p} > T_{n, \alpha, p} \right) \longrightarrow 1$ for any $P_0 \in \s{M}_A$.
\end{thm}

\subsection{Asymptotically correct type I error rate}

Next, we provide conditions under which the proposed test has asymptotically correct type I error rate. We start by introducing an additional condition that we will need.
\begin{description}[style=multiline,leftmargin=1.1cm]
\item[(A4)] Both $\mu_\infty = \mu_0$ and $g_\infty = g_0$, and $r_n :=  \max_v \left|P_0 (\mu_{n,v} - \mu_0)( g_{n,v} - g_0)\right| = \fasterthan\left(n^{-1/2}\right)$.
\end{description}
In concert with Condition (A2), condition (A4) requires that both estimators are consistent. Furthermore, condition (A4) requires that the rate of convergence of the mean of the product of the nuisance errors tend to zero in probability faster than $n^{-1/2}$. We note that $r_n^2 \leq \max_v P_0(\mu_{n,v} - \mu_0)^2 P_0(g_{n,v} - g_0)^2$, so that $r_n$ is bounded by the product of the $L_2(P_0)$ rates of convergence of the two estimators. Therefore, if in particular $\max_v \|\mu_{n,v} - \mu_0\|_{P_0,2} = \fasterthan(n^{-1/4})$ and $\max_v \| g_{n,v} - g_0\|_{P_0,2} = \fasterthan(n^{-1/4})$, then (A4) is satisfied. For example, if $\mu_n$ and $g_n$ are based on correctly-specified parametric models, then (A4) typically holds with room to spare. However, in many contexts, \textit{a priori} correctly specified parametric models for $\mu_0$ and $g_0$ are not available, which motivates the use of semiparametric and nonparametric estimators for $\mu_n$ and $g_n$. While we can expect such semi- or nonparametric estimators to be consistent for a larger class of true functions than parametric estimators, whether (A4) is satisfied depends on the adaptability of the estimators to the specific, often unknown  nature of $\mu_0$ and $g_0$. For this reason, we suggest using ensemble methods based on cross-validation in practice: leveraging several parametric, semiparametric, and nonparametric estimators may give the best chance of minimizing bias and ensuring that (A4) is satisfied. These themes are prevalent throughout the doubly-robust estimation literature, and indeed they are fundamental to nonparametric estimation problems in causal inference (see, e.g.\ \citealp{van2003unified,neugebauer2005prefer, kennedy2016continuous}).

Under these conditions, we have the following result.
\begin{thm}[Weak convergence of $n^{1/2} ( \Omega_n^\circ - \Omega_0)$] \label{thm:weak_conv_omega}
If (A1)--(A2) and (A4) hold, then 
\[\sup_{a \in \s{A}_0} \left| n^{1/2} \left[ \Omega_n^\circ(a) - \Omega_0(a) \right] - n^{1/2} \d{P}_n D_{a, 0}^* \right| \inprob 0 \ ,\]
and in particular, $\left\{ n^{1/2} \left[ \Omega_n^\circ(a)- \Omega_0(a)\right] : a \in \s{A}_0\right\}$ converges weakly as a process in $\ell^\infty(\s{A}_0)$ to a mean-zero Gaussian process $Z_0$ with covariance function given by $\Sigma_0(s,t) := P_0 \left[ D_{s,0}^* D_{t,0}^*\right]$.
\end{thm}
As with the relationship between Theorems~\ref{thm:dr_cons_omega} and~\ref{thm:dr_cons_test}, Theorem~\ref{thm:weak_conv_omega} does not quite imply that the proposed test has asymptotically correct size due to the two additional approximations made in the proposed test. Specifically, it follows from Theorem~\ref{thm:weak_conv_omega} that $P_0\left(\left\|\Omega_n^\circ \right\|_{F_0, p} > T_{0,\alpha,p} / n^{1/2}\right) \longrightarrow \alpha$, where $T_{0,\alpha,p}$ is the $1-\alpha$ quantile of $\|Z_0\|_{F_0,p}$. Validity of the proposed test follows if $\left\|\Omega_n^\circ\right\|_{F_n, p} - \left\|\Omega_n^\circ\right\|_{F_0, p} = \fasterthan(n^{-1/2})$ and $T_{n,\alpha,p} \inprob T_{0,\alpha,p}$. Theorem~\ref{thm:test_size} verifies these facts and concludes that the test has asymptotically valid size.
\begin{thm}[Asymptotic validity of proposed test]\label{thm:test_size}
If conditions (A1)--(A2) and (A4) hold and the distribution function of $\|Z_0\|_{F_0,p}$ is strictly increasing and continuous in a neighborhood of $T_{0,\alpha,p}$, then $P_0\left(n^{1/2}\|\Omega_n^\circ \|_{F_n, p} > T_{n, \alpha, p} \right) \longrightarrow \alpha $ for any $P_0 \in \s{M}_0$.
\end{thm}

\subsection{Asymptotic behavior under local alternatives}

Finally, we demonstrate that the proposed test has power to detect local alternatives approaching the null at the rate $n^{-1/2}$. We let $h : \s{O} \to \d{R}$ be a score function satisfying $P_0h = 0$ and $P_0( h^2)< \infty$. We suppose that the local alternative distributions $P_n$ satisfy
\begin{equation}\lim_{n\to\infty} \int \left[ n^{1/2} \left( dP_n^{1/2} - dP_0^{1/2}\right) - \tfrac{1}{2} h \, dP_0^{1/2} \right]^2 = 0\label{eq:local_alt}\end{equation}
for some $P_0 \in \s{M}_0$. We then have the following result.
\begin{thm}[Weak convergence of $n^{1/2}\Omega_n^\circ$ under local alternatives] \label{thm:weak_conv_local}
If for each $n$, $(O_1, \dotsc, O_n)$ are independent and identically distributed from  $P_n$ satisfying \eqref{eq:local_alt} and the conditions of Theorem~\ref{thm:test_size} hold, then $\left\{ n^{1/2} \Omega_n^\circ(a): a \in \s{A}_0\right\}$ converges weakly under $P_n$ in $\ell^\infty(\s{A}_0)$ to a Gaussian process $\overline{Z}_{0,h}$ with mean $E[\overline{Z}_{0,h}(a)] = P_0 \left( h D_{a,0}^*\right)$ and covariance $\Sigma_0(s,t) := P_0 \left[ D_{s,0}^* D_{t,0}^*\right]$.
\end{thm}
The limiting process $\overline{Z}_{0,h}$ in Theorem~\ref{thm:weak_conv_local} is equal in distribution to $\{ Z_0(a) + P_0 \left( h D_{a,0}^*\right) : a \in \s{A}_0\}$, where $Z_0$ is the limit Gaussian process when generating data under $P_0$ from Theorem~\ref{thm:weak_conv_omega}. 

Theorem~\ref{thm:weak_conv_local} leads to the following local power result for the proposed test.
\begin{thm}[Power under local alternatives]\label{thm:local_alt_power}
If the conditions of Theorem~\ref{thm:weak_conv_local} hold and $T_{0,\alpha, p}$ is the $1-\alpha$ quantile of $\| Z_0 \|_{F_0, p}$, then $P_n\left(n^{1/2}\|\Omega_n^\circ \|_{F_n, p} > T_{n, \alpha, p} \right) \longrightarrow P\left( \|\overline{Z}_{0,h}\|_{F_0, p} >T_{0,\alpha,p} \right) $.
\end{thm}
We note that $P_0\left( \|\overline{Z}_{0,h}\|_{F_0, p} >T_{0,\alpha,p} \right) >\alpha$. Therefore, Theorem~\ref{thm:local_alt_power} implies that, if the two nuisance parameters converge fast enough to the true functions, our test has non-trivial asymptotic power to detect local alternatives approaching the null at the rate $n^{-1/2}$. 

\section{Simulation studies}\label{sec:simulation}

\subsection{Simulation study I: mixed continuous-discrete exposure}

We conducted two simulation studies to examine the finite-sample behavior of the proposed procedure under various null and alternative hypotheses. The general form of our first simulation procedure was as follows. We generated three continuous covariates $W \in \d{R}^3$ from a multivariate normal distribution with mean $(0,0,1)^T$ and identity covariance. In order to generate $A$ given $W$, we define $\lambda_{\beta,\kappa}(w) := \kappa + 2(1 - \kappa) \n{logit}^{-1}\left(\beta^T w - \beta_3\right)$, where $\beta \in \d{R}^3$, $\kappa \in (0, 1)$, and $\n{logit}(x) := \log[x / (1-x)]$, and we define $G_{\beta, \kappa}(u, w) := \lambda_{\beta, \kappa}(w) u + [1-\lambda_{\beta, \kappa}(w)] u^2$ and $G_{\beta,\kappa}^{-1}$ its inverse with respect to the first argument. Finally, we define the mixed discrete-continuous distribution function $F_0$ as $F_0(a) := 0.2\times \left[I_{[0,\infty)}(a) +  I_{[0.5,\infty)}(a) + I_{[1,\infty)}(a)\right] + 0.4 \times B(a; 2,2)$, where $B$ is the distribution function of a beta random variable, and we define $F_0^-$ is the generalized inverse corresponding to $F_0$.  Given $W$, we then simulated $A$ as $F_0^- \circ G_{\beta, \kappa}^{-1}(Z, W)$, where  $Z$ was a Uniform$(0,1)$ random variable independent of $W$, so that $A$ had marginal mass $0.2$ each at $0, 0.5,$ and $1$, and the remaining $0.4$ mass was distributed as $\n{Beta}(2,2)$. For all data generating processes, we set $\kappa = 0.1$ and $\beta = (-1,1,-1)$.

We generated $Y$ given $A$ and $W$ from a linear model with possible interactions and a possible quadratic component. Defining $\mu_{\gamma_1, \gamma_2, \gamma_3}(a, w) :=  \gamma_1^T \bar{w}+  \left( \gamma_2^T \bar{w}\right) a + \gamma_3 (a - 0.5)^2$, where $\bar{w} := (1, w)$, $\gamma_1$ and $\gamma_2$ are elements of $\d{R}^4$, and $\gamma_3 \in \d{R}$, we generated $Y$ from a normal distribution with mean $\mu_{\gamma_1, \gamma_2,\gamma_3}(A, W)$ and variance $ 1 + \left|\mu_{\gamma_1, \gamma_2, \gamma_3}(A,W) \right|$. Given these definitions, we then have $\theta_0(a) = \gamma_{1,1} + \gamma_{1,4} + \left( \gamma_{2,1} + \gamma_{2,4}\right) a + \gamma_3(a - 0.5)^2$. Hence,  $H_0$ holds if and only if $\gamma_{2,1} = - \gamma_{2,4}$ and $\gamma_3 = 0$. We set $\gamma_1 = (0, 2, 2, -2)^T$ for all simulations, and we considered  five combinations of $\gamma_2$ and $\gamma_3$. First, we set $\gamma_2 = (2, 2, 2, -2)^T$ and $\gamma_3 = 0$. We call this the \emph{weak null} because $\mu_0$ depends on $a$ even though $\theta_0$ does not. Second, we simulated data under the strong null by setting $\gamma_2 = (0,0,0,0)^T$ and $\gamma_3 = 0$, so that neither $\mu_0$ nor $\theta_0$ depend on $a$. We also simulated data under four alternative hypotheses. In the first three alternative hypotheses, we set $\gamma_3 = 0$, but varied $\gamma_{2,1} + \gamma_{2,4}$, which is the slope of $\theta_0$. We call these \emph{weak}, \emph{moderate}, and \emph{strong} (linear) alternatives. Finally, we set $\gamma_3 = 2$ and $\gamma_2 = (1,1,-1,-1)^T$, which we call the \emph{quadratic} alternative. These simulation settings are summarized for convenience in Table~\ref{tab:sim_settings}.

\begin{table}
\centering
\begin{tabular}{cccccc}
Setting name &  $\gamma_2$ & $\gamma_3$ & $\|\Omega_0\|_{F_0, 1}$ & $\|\Omega_0\|_{F_0, 2}$ & $\|\Omega_0\|_{F_0, \infty}$ \\
\hline
Weak null  & $(2, 2, 2, -2)^T$ & 0 & 0 &0  & 0\\
Strong null  & $(0, 0, 0, 0)^T$ & 0 & 0& 0 & 0\\
Weak alternative & $(0.5, 1, -1, -0.25)^T$ & 0 & 0.019 & 0.023 & 0.036 \\
Moderate alternative  & $(1, 1, -1, -0.5)^T$ & 0 &0.043 & 0.050 & 0.070 \\
Strong alternative & $(2, 1, -1, -1)^T$  & 0 &0.10 & 0.11 & 0.14\\
Quadratic alternative & $(1, 1, -1, -1)^T$ & 2 & 0.03 & 0.04 & 0.06 \\
\hline
\end{tabular}
\caption{Summary of the six simulation settings used to generate the outcome. We note that $\gamma_2 =  (0, 2, 2, -2)^T$ for all settings. For context, $\n{Var}(Y) \in (4,4.5)$ for all alternative simulation settings.}
\label{tab:sim_settings}
\end{table}

For each sample size $n \in \{100, 250,500, 750, 1000, 2500, 5000\}$ and each of the settings listed in Table~\ref{tab:sim_settings}, we generated 1000 datasets using the process described above. For each dataset, we estimated the pair of nuisance parameters $(\mu_n, g_n)$ in the following ways. First, we estimated $\mu_n$ using a correctly specified linear regression, and $g_n$ using maximum likelihood estimation with a correctly specified parametric model for $\beta$ with $\kappa$ set to the true data-generating value. Second, we used the same correctly-specified procedure for $g_n$, but used an incorrectly specified linear regression to estimate $\mu_n$ by excluding the interactions between $A$ and $W$ and $W_3$ from the regression. Third, we used the correctly specified linear regression to estimate $\mu_n$, but used an incorrectly specified parametric model for $g_n$ by maximizing the incorrectly specified likelihood $(\alpha_1, \alpha_2) \mapsto \sum_{i=1}^n \log \left\{2U_i + \left(1 - 2U_i\right) \n{logit}^{-1}\left(\alpha_1 W_1 + \alpha_2 W_2\right)\right\}$ for $U_i = F_n(A_i)$. Fourth, we used the incorrectly specified parametric models for both $\mu_n$ and $g_n$. Fifth, we estimated $\mu_n$ and $g_n$ nonparametrically. To estimate $\mu_n$ nonparametrically, we used SuperLearner \citep{vanderlaan2007super} with a library consisting of linear regression, linear regression with interactions, a generalized additive model, and multivariate adaptive regression splines. To estimate $g_n$ nonparametrically, we used an adaptation of the method described in \cite{diaz2011super} that allows for mass points. For each of these five pairs of estimation strategies for $\mu_n$ and $g_n$, we used the method described in this article with $p \in \{1, 2, \infty\}$ to test the null hypothesis. For nonparametric nuisance estimation, we used both the cross-fitted estimator $\Omega_n^\circ$ and the non-cross-fitted estimator $\Omega_n$ in order to assess the effect of cross-fitting on type I error. Finally, we compared our test to a test based on dichotomizing $A$. Specifically, we defined $\bar{A} := I_{[0.5,1]}(A)$, and used Targeted Minimum-Loss based Estimation (TMLE) \citep{vanderlaan2011tmle} to test the null hypothesis that $E_0[E_0(Y \mid \bar{A} = 0, W)] = E_0[E_0(Y \mid \bar{A} = 1, W)]$. We used cross-fitted SuperLearners with the same library as above as the nuisance estimators for TMLE.

\begin{figure}[t!]
\centering
\includegraphics[width=\linewidth]{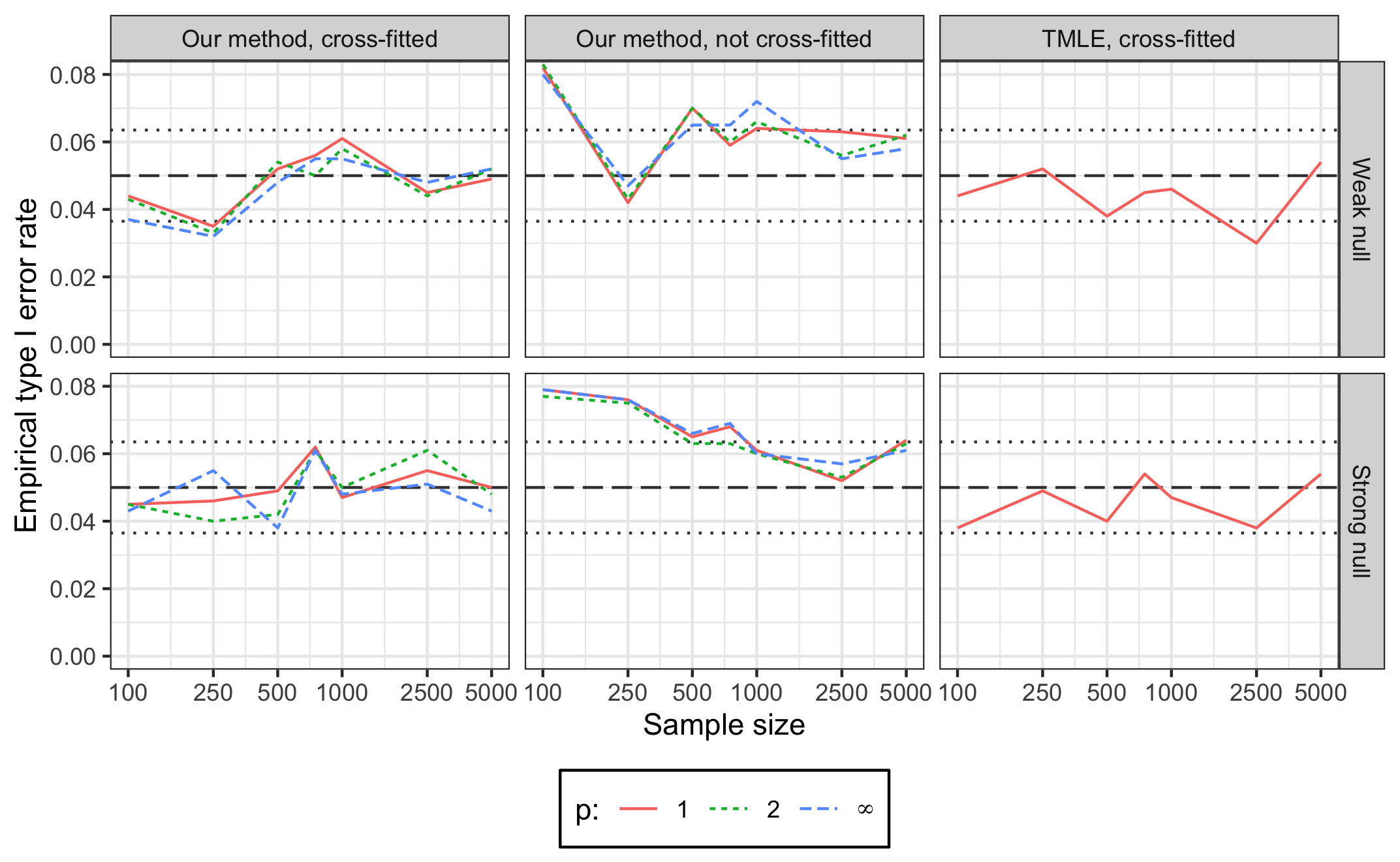}
\caption{Empirical type I error rate of nominal $\alpha = 0.05$ level tests using the nonparametric nuisance estimators. The left column is our cross-fitted test with nonparametric nuisances estimators, the middle column is the same without cross-fitting, and the right column is a TMLE-based test using a dichotomized exposure. Horizontal wide-dash line indicates then nominal 0.05 test size, and horizontal dotted lines indicate expected sampling error bounds were the true size 0.05.}
\label{fig:sim_size_np}
\end{figure}

We now turn to the results of the simulation study. We focus on the results from the use of nonparametric nuisance estimators, since this is what we suggest to use in practice. The results from the use of parametric nuisance estimators were in line with expectations based on our theoretical results; full details of the results may be found in Supplementary Material. Figure~\ref{fig:sim_size_np} displays the empirical type I error rate for the three estimators with nonparametric nuisance estimators. Our tests with cross-fitted nuisance estimators (first column) had empirical error rates within Monte Carlo error of the nominal error rate at all sample sizes and under both the strong and weak nulls. This empirically validates the large-sample theoretical guarantee of Theorem~\ref{thm:test_size}, and also indicates that the type I error of the method is valid even for small sample sizes. However, the nonparametric nuisance estimators without cross-fitting (second column) had type I error significantly larger than $0.05$ for $n \leq 1000$. This suggests that the cross-fitting procedure reduced the bias of the estimator of $\Omega_0$ and/or of the bias of the estimator of the quantile $T_{0,\alpha,p}$ for small and moderate sample sizes, resulting in improved type I error rates. The TMLE-based test with a dichotomized exposure also had empirical error rates within Monte Carlo error of the nominal rate for all sample sizes under both types of null hypotheses, as expected.

Figure~\ref{fig:sim_power_np} displays the empirical power using the nonparametric nuisance estimators. We omitted the estimator without cross-fitting, since this estimator had poor type I error control. Power was generally very low for $n = 100$, but increased with sample size and with distance from the null. For the three linear alternatives (first three rows), our test had only slightly (i.e.\ 5-10 percentage points) better power than the TMLE-based test using a dichotomized exposure.  This makes sense, since the true effect size induced by dichotomization of the exposure increased with the slope of $\theta_0$ in the case that $\theta_0$ was linear. However, for the quadratic alternative, the test proposed here had substantially larger power than the TMLE-based test. For example, at sample size $n=1000$, the TMLE-based test had power 0.09, while our test had power between 0.25 and 0.35, and at sample size $n=5000$, the TMLE-based test had power 0.25, while our test had power near 0.85. This can be explained by the fact that the true effect size induced by dichotomization for the quadratic alternative was close to zero because the axis of symmetry for the parabolic effect curve was 0.5, the same as the point of dichotomization. This suggests that, as has been previously noted (e.g.\ \citealp{fedorov2009consequence}), dichotomization can result in substantial loss of power for certain types of data-generating mechanisms. Discretizing the exposure into more categories would increase the power of the TMLE-based test, but in practice it is hard to know what discretization will yield acceptable power without knowing the true form of $\theta_0$.

\begin{figure}[t!]
\centering
\includegraphics[width=\linewidth]{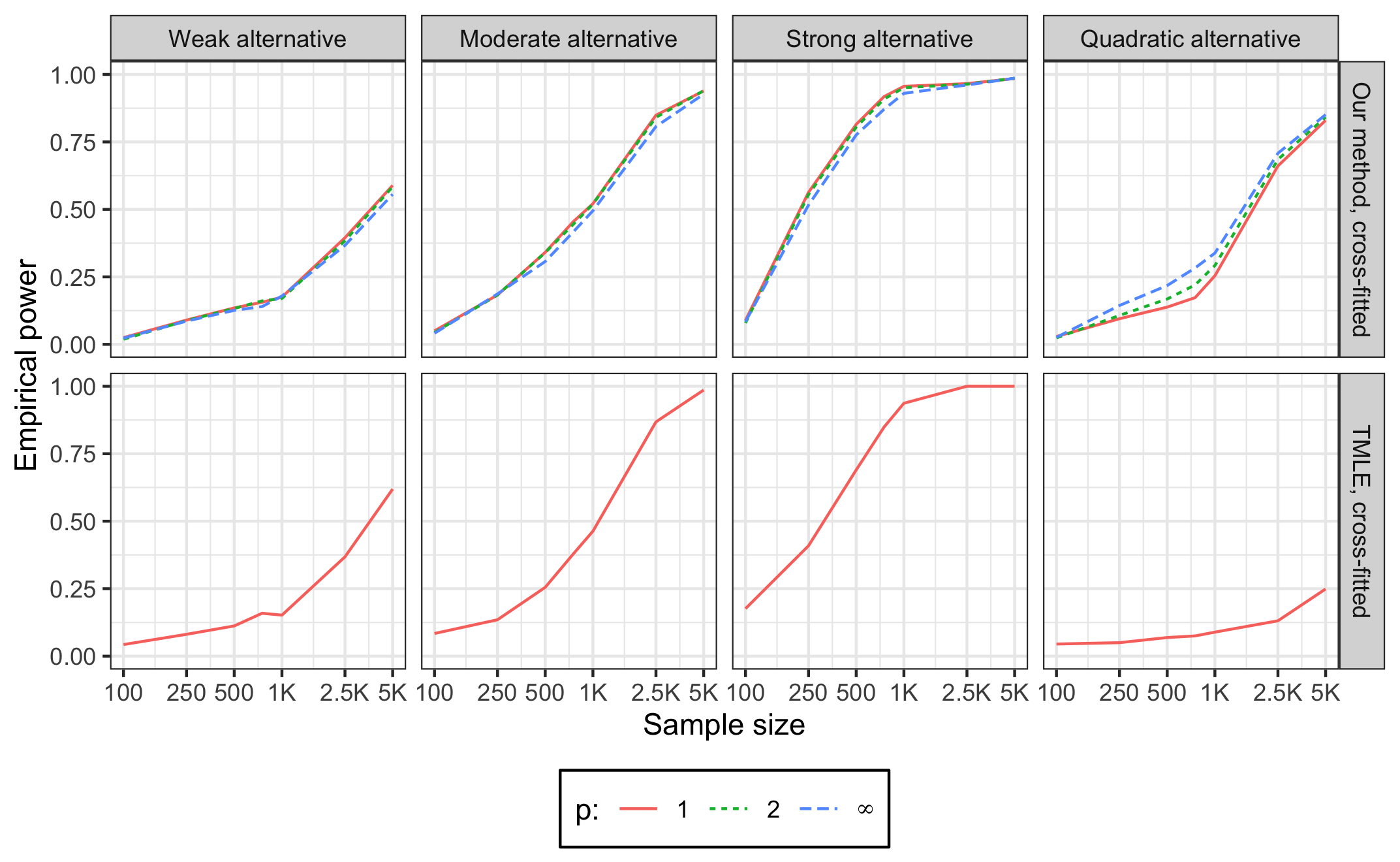}
\caption{Empirical power of nominal $\alpha = 0.05$ level tests using nonparametric nuisance estimators. Columns indicate the alternative hypothesis used to generate the data, and rows indicate the method used to test the null hypothesis.}
\label{fig:sim_power_np}
\end{figure}

Overall, we observed little systematic difference in type I error rates between the three values of $p$ using either type of nuisance estimator for our test. For the linear alternatives, the test with $p=\infty$ had consistently slightly smaller power than that with $p=1$ or $p=2$. However, for the quadratic alternative, the test with $p=\infty$ had consistently larger power than the others. Therefore, which value of $p$ yields the greatest power depends on the shape of the true effect curve. As noted previously, we recommend using $p=\infty$ in practice due in part to the improved power against nonlinear alternatives.

\subsection{Simulation study II: discrete exposure with many levels}

In a second simulation study, we assessed the effect of increasing the number of levels of a discrete exposure on the properties of tests of the null hypothesis of no average causal effect. We simulated three covariates $W \in \d{R}^3$ from independent uniform distributions on $[-1,1]$, $[-1,1]$, and $[0,2]$, respectively. For a number of levels $k$, we set $P_0(A = a \mid W = w) = \n{logit}^{-1}[(.5 - a)(\beta^T w) ] / h(w)$ for $a \in \{1/k, 2/k,\dotsc, 1\}$, and $P_0(A = a \mid W = w) = 0$ otherwise, where $h(w)$ is a normalizing constant. Given $A$ and $W$, we simulated $Y$ as in simulation study I described above. For each $n \in \{250, 500, 750\}$, we considered eight values of $k$ between $k = 5$ and $k = n/2$, which allowed us to assess the effect of the number of discrete components of the exposure on the properties of testing procedures. For each setting, we simulated 1000 datasets,  and used two methods to test the null hypothesis of no average causal effect of $A$ on $Y$. First, we used the method described here with $p \in \{1,2,\infty\}$. Second, we used a chi-squared test based on an augmented IPW  (AIPW) estimator of $\theta_0(a_j)$ for each $a_j$ in the support of $A$. Since the exposure was discrete, the AIPW-based test had \emph{asymptotically} valid size for any fixed $k$; our goal was to assess its finite-sample performance as a function of $k$.  For both tests, we used cross-fitted maximum likelihood estimators from correctly-specified parametric models for the nuisance estimators.

\begin{figure}[ht!]
\centering
\includegraphics[width=\linewidth]{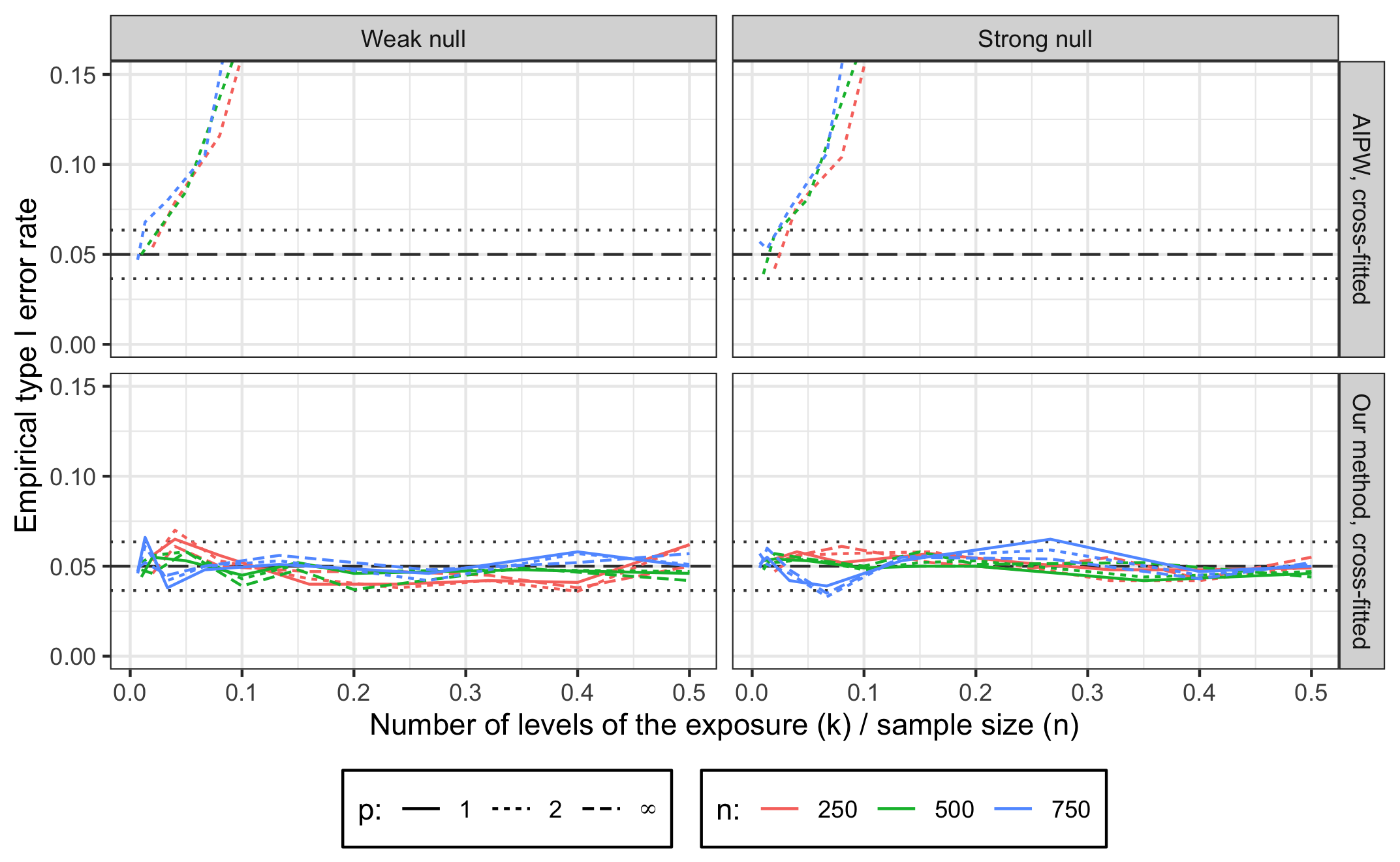}
\caption{Empirical type I error rate of nominal $\alpha = 0.05$ level tests in the second simulation study. The $x$-axis is the ratio of $k$, the number of levels of the exposure, to $n$, the sample size. Columns indicate the null hypothesis used to generate the data, and rows indicate the method used to test the null hypothesis.}
\label{fig:sim_size_disc}
\end{figure}

Figure~\ref{fig:sim_size_disc} displays the empirical type I error of the two methods under the weak and strong null hypotheses. Our methods (bottom row) had type I error near 0.05 for all $n$ and $k$ considered. AIPW, on the other hand, only had valid type I error for the smallest value of $k$ considered ($k = 5$). As the number of levels of $A$ grew, the type I error rate of the AIPW-based test rapidly grew to 1. In Supplementary Material, we display the power of our test, which was constant in $k$, and for all settings considered increased with $n$. Therefore, use of the methods proposed here is not limited to exposures with a continuous component, but should be considered for all discrete exposures with more than a few values.


\section{BMI and T-cell response in HIV vaccine studies}\label{sec:bmi}

Numerous scientific studies have found a negative association between obesity or BMI and immune responses to vaccination. In \cite{jin2015multiple}, the authors found that low BMI ($<25$) participants in early-phase HIV vaccine trials had a higher rate of CD4+ T cell response than high BMI ($\geq 30$) participants. They also found a significant effect of BMI in a logistic regression of CD4+ responses on sex, age, BMI, vaccination dose, and number of vaccinations (OR: 0.92; 95\% CI: 0.86--0.98; $p$=0.007). However, as discussed in the introduction, this odds ratio only has a causal interpretation under strong parametric assumptions.

In \cite{westling2020isotonic}, the authors estimated the causal dose-response curve $\theta_0$, adjusting for the same set of confounders as did \cite{jin2015multiple}, under the assumption that it is decreasing. However, they did not assess the null hypothesis that the curve was flat. We note that \cite{westling2020isotonic} estimated $\theta_0(20) - \theta_0(35)$, i.e.\ the difference between the probabilities of having a positive CD4+ immune response under assignment to BMIs of 20 and 35, to be 0.22, with 95\% confidence interval 0.03--0.41. However, this confidence interval is only valid under the assumption that $\theta_0(20) > \theta_0(35)$, and hence cannot be used as evidence against the null hypothesis that $\theta_0$ is flat. Furthermore, the fact that the lower end of this confidence interval is relatively close to zero suggests that there may not actually be strong evidence against this null. Here, we formally assess this null hypothesis using the same data as \cite{westling2020isotonic}.

The data consist of pooled vaccine arms from 11 phase I/II clinical trials conducted through the HIV Vaccine Trials Network (HVTN). Descriptions of these trials may be found in \cite{jin2015multiple} and \cite{westling2020isotonic}. CD4+ and CD8+ T-cell responses at the first visit following administration of the last vaccine dose were measured using validated intracellular cytokine staining, and these continuous responses were converted to binary indicators of whether there was a significant change from baseline using the method described in \cite{jin2015multiple}. After excluding three participants with missing BMI and participants with missing immune response, our analysis datasets consisted of a total of $n=439$ participants for the analysis of CD4+ responses and $n=462$ participants for CD8+ responses.

We tested the null hypotheses that there is a causal effect of BMI on CD4+ and CD8+ T-cell responses using the method developed in this paper with $p = \infty$ and $V = 10$ folds. To estimate the outcome regression $\mu_0$ and the propensity score $g_0$, we used SuperLearner \citep{vanderlaan2007super, diaz2011super} with flexible libraries consisting of generalized linear models, generalized additive models, multivariate regression splines, random forests, and gradient boosting. For the analysis of the effect of BMI on CD4+ responses, we found $p = 0.16$, and for the analysis of the effect of BMI on CD8+ responses, we found $p = 0.22$. Hence, we do not find evidence of a causal effect of BMI on the probability of having a positive immune response in these data. Plots of $\Omega_n$ are presented in Supplementary Material.

\section{Discussion}\label{sec:discussion}

We have presented a nonparametric method for testing the null hypothesis that a causal dose-response curve is flat, for use in observational studies with no unobserved confounding. The key idea behind our test was to translate the null hypothesis on the parameter of interest, which is not a pathwise differentiable parameter in the nonparametric model, into a null hypothesis on a primitive parameter, which is pathwise differentiable.

In addition to permitting the use of methods and theory for pathwise differentiable parameters, using the primitive function gives our tests non-trivial power to detect alternatives approaching the null at the rate $n^{-1/2}$.  However, results from the literature concerning tests of marginal regression functions suggest that any test able to detect $n^{-1/2}$-rate alternatives must necessarily have low finite-sample power against certain smooth alternatives \citep{horowitz2001adaptive}. Analogous results regarding causal dose-response functions are not to our knowledge available, but we conjecture that appropriate parallels can be established. Therefore, our tests may have poor finite-sample power against certain shapes of dose-response functions; in particular, functions that are nearly flat except for a sharp peak in a narrow range of the support of the exposure. We suggest that users conduct numerical studies on the power of our tests if they expect that their dose-response function may look like this. Tests based on an estimator of the dose-response function, which to our knowledge do not yet exist, may not have this problem, although such tests would also have slower than $n^{-1/2}$ rates of convergence. We leave further inquiry along these lines to future work.

An additional benefit of using the primitive function defined here is that it makes the test agnostic to the marginal distribution of the exposure: the test works equally well with discrete, continuous, and mixed discrete-continuous exposures. This was validated in numerical studies, where in particular we demonstrated that a traditional doubly-robust test of no average causal effect in the setting of a fully discrete exposure quickly became invalid as the number of levels of the exposure grew. We note that tests based on directly estimating the dose-response function, such as tests based on the local linear estimator, may only work in the context of fully continuous exposures.

Several modifications of the proposed test may be of interest in future research. Here, we studied the properties of the test for fixed values of $p$. In numerical studies, we found little difference in the performance of the test for $p \in \{1,2,\infty\}$, and we do not expect that the choice of $p$ would drastically change the results in most cases. However, the results presented herein were for fixed values of $p$, and so if a researcher were to select a value of $p$ based on the results of the test, the test may no longer have asymptotically valid type I error. In future research, it would be of interest to adaptively select a value of $p$ to maximize power while retaining type I error control. In addition, here, we used the empirical distribution function as our weight function to assess whether the primitive parameter is flat. Alternative weight functions could be used to, for instance, place more emphasis in the tails or center of the distribution of the exposure, or a weight function could possibly be adaptively chosen to maximize power. Finally, while we used a one-step estimator of the primitive parameter, TMLE could be used instead.


\clearpage

\begin{center}
\Large
\textit{
\textsc{Supplementary Material for:}\\
\vspace{.5em}
Nonparametric tests of the causal null \\with non-discrete exposures}
\end{center}

\vspace{1em}

\section*{Proof of Theorems}

\begin{proof}[\bfseries{Proof of Proposition~\ref{prop:equivalence}}]
(1)$\implies$(2): Let $a \in \s{A}_0$. Then, since $\theta_0(u) = \theta_0(a)$ for all $u \in \s{A}_0$, $\gamma_0 = \int\theta_0(u) \, dF_0(u) =  \int\theta_0(a) \, dF_0(u) = \theta_0(a)$. 

(2)$\implies$(1): trivial.

(2)$\implies$(3): Let $a \in \d{R}$. Then $\Gamma_0(a) =\int_{-\infty}^a \theta_0(u) \, dF_0(u) = \gamma_0F_0(a)$, so $\Omega_0(a) = 0$.

(3)$\implies$(2): We proceed by contradiction:  suppose that $\theta_0(a) \neq \gamma_0$ for some $a \in \s{A}_0$. We assume first that $\theta_0(a)  - \gamma_0 = \delta > 0$. Then since by assumption $\theta_0$ is continuous on $\s{A}_0$, there exists $\varepsilon > 0$ such that $|\theta_0(u) - \theta_0(a) | \leq \delta / 2$ for all $u \in \s{A}_0 \cap [a - \varepsilon, a + \varepsilon]$, which implies that $\theta_0(u) -\gamma_0 \geq \delta / 2$ for all such $u$. We then have 
\begin{align*}
0 = \Omega_0(a + \varepsilon) - \Omega_0(a - \varepsilon) &= \int_{a - \varepsilon}^{a+\varepsilon} [\theta_0(u) - \gamma_0] \, dF_0(u) = \int_{(a - \varepsilon, a+\varepsilon] \cap\s{A}_0} [\theta_0(u) - \gamma_0] \, dF_0(u) \\
&\geq (\delta / 2) \left[F_0(a + \varepsilon) - F_0(a - \varepsilon)\right] > 0 \ ,
\end{align*}
where the last inequality follows because $a$ is in the support of $F_0$. This is a contradiction, and therefore $\theta_0(a) \leq \gamma_0$. The argument if $\theta_0(a) < \gamma_0$ is essentially identical, and since $a \in \s{A}_0$ was arbitrary, this yields that $\theta_0(a) = \gamma_0$ for all $a \in \s{A}_0$. 

(3)$\implies$(4): trivial.

(4)$\implies$(3): We proceed again by contradiction. Suppose $|\Omega_0(a)| > 0$ for some $a \in \d{R}$. First suppose $a \in \s{A}_0$. If $a$ is a mass point of $F_0$, then clearly $\|\Omega_0\|_{F_0,1} >0$, a contradiction. If $a$ is not a mass point of $F_0$, then for any $\varepsilon > 0$
\[\left|\Omega_0(a + \varepsilon) - \Omega_0(a)\right| \leq \int_{a}^{a + \varepsilon} \left| \theta_0(u) - \gamma_0\right| \, dF_0(u) \leq c[F_0(a + \varepsilon) - F_0(a)]  \]
for $c < \infty$, and since $F_0(a + \varepsilon) \to F_0(a)$ as $\varepsilon \to 0$, $\Omega_0$ is right-continuous at $a$. An analogous argument shows that $\Omega_0$ is also left-continuous at $a$. This implies that $|\Omega_0|$ is positive in a neighborhood of $a$, which implies since $a \in \s{A}_0$ that $\|\Omega_0\|_{F_0,1} >0$, a contradiction. Finally, if $a \in \d{R}$ is not an element of $\s{A}_0$, then $\Omega_0(a) = \Omega_0(a_0)$ for $a_0 := \sup\{ u \in\s{A}_0 : u < a\}$, so that $|\Omega_0(a)| > 0$ implies $|\Omega_0(a_0)| > 0$, and since $a_0 \in \s{A}_0$ (because $\s{A}_0$ is closed), this leads to a contradiction.

\end{proof}

\clearpage

\begin{proof}[\bfseries{Proof of Proposition~\ref{prop:eif}}]
We let $\left\{P_\varepsilon : |\varepsilon | \leq \delta \right\}$ be any one-dimensional differentiable in quadratic mean (DQM) path in $\s{M}^\circ$ such that $P_{\varepsilon = 0} = P_0$, where $\s{M}^\circ$ be the subset of $P \in \s{M}$ such that there exists $\eta > 0$ such that $ g_P(a,w) \geq \eta$ for $P$-a.e.\ $(a,w)\}$ and $E_P[Y^2] < \infty$. We let $(y, a, w) \mapsto \dot\ell_0(y, a, w)$ be the score function of the path at $\varepsilon = 0$, We note that $P_0\dot\ell_0 = 0$ and $P_0 \dot\ell_0^2 < \infty$ \citep{bickel1998efficient}. Furthermore, we define $\dot\ell_0(a, w) := E_0[ \dot\ell_0 \mid A = a, W =w]$ and analogously $\dot\ell_0(a)$ and $\dot\ell_0(w)$ as  the marginal score functions and $\dot\ell_0(y \mid a, w) := \dot\ell_0(y, a, w) - \dot\ell_0(a,w)$ as the conditional score function, which has mean zero conditional given $A, W$ under $P_0$.

The nonparametric efficient influence function $D_{a_0,0}^*$ of $\Omega_0(a_0)$ can be derived by showing that $\left.\frac{\partial}{\partial\varepsilon} \Omega_\varepsilon(a_0) \right|_{\varepsilon = 0} = P_0 \left( D_{a_0,0}^* \dot\ell_0 \right)$. We have by definition of $\Omega_0$ and ordinary rules of calculus that
\begin{align}
\left.\frac{\partial}{\partial\varepsilon} \Omega_\varepsilon(a_0) \right|_{\varepsilon = 0} &= \left.\frac{\partial}{\partial\varepsilon} \Gamma_\varepsilon(a_0) \right|_{\varepsilon = 0} - \left.\frac{\partial}{\partial\varepsilon} \gamma_\varepsilon \right|_{\varepsilon = 0} F_0(a_0) - \gamma_0 \left.\frac{\partial}{\partial\varepsilon} F_\varepsilon(a_0) \right|_{\varepsilon = 0}. \label{eq:inf_decomp1}
\end{align}
We first note that since $F_P(a_0) = E_P[ I(A \leq a_0)] = \int I(a \leq a_0) \, dP(Y, A, W)$ for any $P$,
\begin{align*}
 \left.\frac{\partial}{\partial\varepsilon} F_\varepsilon(a_0) \right|_{\varepsilon = 0} &= \left.\frac{\partial}{\partial\varepsilon} \int I(a \leq a_0) dP_\varepsilon(y,a,w) \right|_{\varepsilon = 0} = \int I(a \leq a_0) \dot\ell_0(y, a, w) dP_0(y,a,w) \\
 &= E_0 \left[  I(A \leq a_0) \dot\ell_0(Y,A,W)\right].
 \end{align*}
 Therefore, the term $\gamma_0 \left.\frac{\partial}{\partial\varepsilon} F_\varepsilon(a_0) \right|_{\varepsilon = 0}$ in~\eqref{eq:inf_decomp1} contributes $-\gamma_0  I(a \leq a_0)$ to the uncentered influence function.
 
Next, by definitions of $\Gamma_P$ and $\gamma_P$ and the product rule,
\begin{align*}
 \left.\frac{\partial}{\partial\varepsilon} \Gamma_\varepsilon(a_0) \right|_{\varepsilon = 0} - \left.\frac{\partial}{\partial\varepsilon} \gamma_\varepsilon \right|_{\varepsilon = 0} F_0(a_0) &= \left.\frac{\partial}{\partial\varepsilon}  \iint  I(a \leq a_0) \mu_\varepsilon (a, w) \, dF_\varepsilon(a) \, dQ_\varepsilon(w)  \right|_{\varepsilon = 0}  \\
 &\qquad - \left.\frac{\partial}{\partial\varepsilon} \iint  \mu_\varepsilon (a, w) \, dF_\varepsilon(a) \, dQ_\varepsilon(w) \right|_{\varepsilon = 0} F_0(a_0) \\
 &=\left.\frac{\partial}{\partial\varepsilon}  \iint  \left[ I(a \leq a_0)  - F_0(a_0) \right] \mu_\varepsilon (a, w) \, dF_\varepsilon(a) \, dQ_\varepsilon(w)  \right|_{\varepsilon = 0} \\
 &= \left.\frac{\partial}{\partial\varepsilon}  \iiint  \left[ I(a \leq a_0)  - F_0(a_0) \right] y \, dP_\varepsilon(y \mid a, w) \, dF_0(a) \, dQ_0(w)  \right|_{\varepsilon = 0} \\
 &\qquad + \left.\frac{\partial}{\partial\varepsilon}  \iint  \left[ I(a \leq a_0)  - F_0(a_0) \right] \mu_0 (a, w) \, dF_\varepsilon(a) \, dQ_0 (w)  \right|_{\varepsilon = 0} \\ 
 &\qquad + \left.\frac{\partial}{\partial\varepsilon}  \iint  \left[ I(a \leq a_0)  - F_0(a_0) \right] \mu_0 (a, w) \, dF_0(a) \, dQ_\varepsilon(w)  \right|_{\varepsilon = 0}.
\end{align*}
Now using basic properties of score functions, we have
\begin{align*}
& \left.\frac{\partial}{\partial\varepsilon}  \iiint  \left[ I(a \leq a_0)  - F_0(a_0) \right] y \, dP_\varepsilon(y \mid a, w) \, dF_0(a) \, dQ_0(w)  \right|_{\varepsilon = 0} \\
&\qquad =    \iiint  \left[ I(a \leq a_0)  - F_0(a_0) \right] y \dot\ell_0(y \mid a, w) \, dP_0(y \mid a, w) \, dF_0(a) \, dQ_0(w)  \\
&\qquad = E_0 \left\{\left[ I(A \leq a_0)  - F_0(a_0) \right] \frac{Y}{g_0(A, W)} \dot\ell_0(Y \mid A, W) \right\} \\
&\qquad = E_0 \left\{\left[ I(A \leq a_0)  - F_0(a_0) \right] \frac{Y}{g_0(A, W)} \left[ \dot\ell_0(Y, A, W)  - \dot\ell_0(A, W) \right] \right\} \\
&\qquad =  E_0 \left\{\left[ I(A \leq a_0)  - F_0(a_0) \right] \frac{Y}{g_0(A, W)} \dot\ell_0(Y, A, W)\right\} \\
&\qquad\qquad -E_0 \left\{\left[ I(A \leq a_0)  - F_0(a_0) \right] \frac{E_0 \left[Y \mid A, W \right]}{g_0(A, W)} \dot\ell_0(A, W) \right\} \\
&\qquad = E_0 \left\{\left[ I(A \leq a_0)  - F_0(a_0) \right] \frac{Y}{g_0(A, W)} \dot\ell_0(Y, A, W)\right\} \\
&\qquad\qquad -E_0 \left\{\left[ I(A \leq a_0)  - F_0(a_0) \right] \frac{\mu_0(A,W)}{g_0(A, W)} \dot\ell_0(Y, A, W) \right\} \\
&\qquad = E_0 \left\{\left[ I(A \leq a_0)  - F_0(a_0) \right] \frac{Y - \mu_0(A, W)}{g_0(A, W)} \dot\ell_0(Y, A, W)\right\}
\end{align*}
since $E_0[ h(A, W) \dot\ell_0(A, W)] = E_0[ h(A, W) \dot\ell_0(Y,A, W)]$ for any suitable $h$. Note that we have used the assumption that $g_0$ is almost surely positive in the above derivation. We also have
\begin{align*}
& \left.\frac{\partial}{\partial\varepsilon}  \iint  \left[ I(a \leq a_0)  - F_0(a_0) \right] \mu_0 (a, w) \, dF_\varepsilon(a) \, dQ_0 (w)  \right|_{\varepsilon = 0} \\
&\qquad =   \iint  \left[ I(a \leq a_0)  - F_0(a_0) \right] \mu_0 (a, w)  \dot\ell_0(a) \, dF_0(a) \, dQ_0 (w) \\
&\qquad =\int  \left[ I(a \leq a_0)  - F_0(a_0) \right] \theta_0(a)  \dot\ell_0(a) \, dF_0(a)\\
&\qquad = E_0 \left\{  \left[ I(A \leq a_0)  - F_0(a_0) \right]  \theta_0(A) \dot\ell_0(A)\right\} \\
&\qquad = E_0 \left\{  \left[ I(A \leq a_0)  - F_0(a_0) \right]  \theta_0(A) \dot\ell_0(Y,A,W)\right\}
\end{align*}
by an analogous argument, and  similarly
\begin{align*}
& \left.\frac{\partial}{\partial\varepsilon}  \iint  \left[ I(a \leq a_0)  - F_0(a_0) \right] \mu_0 (a, w) \, dF_0(a) \, dQ_\varepsilon (w)  \right|_{\varepsilon = 0} \\
&\qquad =   \iint  \left[ I(a \leq a_0)  - F_0(a_0) \right] \mu_0 (a, w)  \dot\ell_0(w) \, dF_0(a) \, dQ_0 (w) \\
&\qquad = E_0 \left\{  \int \left[ I(a \leq a_0)  - F_0(a_0) \right] \mu_0(a, W) \, dF_0(a) \dot\ell_0(W)\right\} \\
&\qquad = E_0 \left\{  \int \left[ I(a \leq a_0)  - F_0(a_0) \right] \mu_0(a, W) \, dF_0(a) \dot\ell_0(Y,A,W)\right\}.
\end{align*}
Putting it together, the uncentered influence function is 
\begin{align*}
\left[I(a \leq a_0)  - F_0(a_0) \right] \left[\frac{y - \mu_0(a, w)}{g_0(a, w)} +\theta_0(a) \right] +  \int \left[ I(u \leq a_0)  - F_0(a_0) \right] \mu_0(u, w) \, dF_0(u) - \gamma_0 I(a \leq a_0).
\end{align*}
Influence functions have mean zero, so we need to subtract the mean under $P_0$ of this uncentered function to obtain the centered influence function.
The mean under $P_0$ of this function is $2\Omega_0(a_0) + \gamma_0 F_0(a_0)$, so the centered influence function is
\begin{align*}
&\left[I(a \leq a_0)  - F_0(a_0) \right] \left[\frac{y - \mu_0(a, w)}{g_0(a, w)} +\theta_0(a) - \gamma_0\right] \\
&\qquad +  \int \left[ I(u \leq a_0)  - F_0(a_0) \right] \mu_0(u, w) \, dF_0(u) -2\Omega_0(a_0),
\end{align*}
which equals $D_{a_0,0}^*(y,a,w)$ as claimed. When $g_0$ is almost surely bounded away from zero and $E_0[Y^2] < \infty$, this function has finite variance.
\end{proof}

\clearpage

Before proving our main results, we derive a first-order expansion of $\Omega_n^\circ(a)$. We define 
\begin{align*}
D_{a_0,n,v}(y, a, w) &:= \left[ I_{(-\infty, a_0]}(a) - F_{n,v}(a_0)\right] \left[  \frac{ y - \mu_{n,v}(a,w)}{g_{n,v}(a,w)}+ \theta_{\mu_{n,v}, Q_{n,v}}(a)-\gamma_{\mu_{n,v}, F_{n,v}, Q_{n,v}}\right]\\
&\qquad + \int \left[ I_{(-\infty, a_0]}(\tilde{a}) - F_0(a_0)\right] \mu_{n,v}(\tilde{a}, w) F_0(d\tilde{a}) - \Omega_{\mu_{n,v}, F_0, Q_{n,v}}(a_0) \ ,\\
D_{a_0,\mu, g}(y, a, w) &:= \left[ I_{(-\infty, a_0]}(a) - F_0(a_0)\right] \left[\frac{ y - \mu(a,w)}{g(a,w)}+\theta_{\mu, Q_0}(a)-\gamma_{\mu, F_0, Q_0}\right] \\
& \qquad + \int \left[ I_{(-\infty, a_0]}(\tilde{a}) - F_0(a_0)\right] \mu(\tilde{a}, w) F_0(d\tilde{a}) - \Omega_{\mu, F_0, Q_0}(a_0)\\
D_{a_0,n,v}^{\circ}(y,a,w) &:= \left[ I_{(-\infty, a_0]}(a) - F_{n,v}(a_0)\right] \left[  \frac{ y - \mu_{n,v}(a,w)}{g_{n,v}(a,w)} +\theta_{\mu_{n,v}, Q_0}(a) \right]  \\
&\qquad + \int \left[ I_{(-\infty, a_0]}(\tilde{a}) - F_0(a_0)\right] \mu_{n,v}(\tilde{a}, w) F_0(d\tilde{a}) \\
D_{a_0, \infty}^{\circ}(y,a,w) &:= \left[ I_{(-\infty, a_0]}(a) - F_0(a_0)\right] \left[ \frac{ y - \mu_\infty(a,w)}{g_\infty(a,w)} + \theta_{\mu_{\infty}, Q_0}(a)\right] \\
&\qquad + \int \left[ I_{(-\infty, a_0]}(\tilde{a}) - F_0(a_0)\right] \mu_\infty(\tilde{a}, w) F_0(d\tilde{a}) \\
D_{a_0,\infty} &:= D_{a_0,\mu_\infty, g_\infty} \\
D_{a_0,\infty}^* &:=D_{a_0, \infty} - \Omega_0(a_0)
\end{align*}
and
\begin{align*}
R_{n,a_0,v,1} &:= (\d{P}_{n,v} - P_0) \left( D_{a_0,n,v}^{\circ} - D_{a_0, \infty}^{\circ} \right)\ ,\\
R_{n,a_0,v,2} &:=  \left( \gamma_{\mu_{n,v}, F_0, Q_0} - \gamma_{\mu_{n,v}, F_0, Q_{n,v}}- \gamma_{\mu_0, F_0, Q_0}+ \gamma_{\mu_\infty, F_0, Q_0}   \right) \left[ F_{n,v}(a_0) - F_0(a_0) \right] \\
R_{n,a_0,v,3} &:= \iint \left[ I_{(-\infty, a_0]}(a) - F_{n,v}(a_0)\right]  \mu_{n,v}(a, w) \, (F_{n,v} - F_0)(da) \, (Q_{n,v} - Q_0)(dw)\\
R_{n,a_0,v,4} &:= \iint\left[ I_{(-\infty, a_0]}(a) - F_{n,v}(a_0)\right] \left[\mu_{n,v}(a, w) - \mu_0(a, w)\right] \left[ 1 - \frac{g_0(a, w)}{g_{n,v}(a,w)}\right] \, F_0(da) \, Q_0(dw)\\
R_{n,a_0,v,5} &:=  \left( 1 - \frac{VN_v}{n} \right) \d{P}_{n,v} D_{a_0, \infty}^*\ .
\end{align*}

\begin{lemma}[First-order expansion of estimator]\label{lemma:first_order}
If condition (A3) or (A4) holds, then $\Omega_{n}^\circ(a_0) -  \Omega_0(a_0) = \d{P}_nD_{a_0, \infty}^* + \frac{1}{V}\sum_{v=1}^V \sum_{j=1}^5 R_{n,a_0,v,j}.$
\end{lemma}
\begin{proof}[\bfseries{Proof of Lemma~\ref{lemma:first_order}}]
We recall that $\theta_{\mu, Q}(a) := \int \mu(a, w) \, Q(dw)$, 
\[ \Omega_{\mu, F, Q}(a) := \iint \left[ I_{(-\infty, a_0]}(a) - F(a_0)\right] \mu(a, w) \, F(da) \, Q(dw) \ ,\]
 and $\gamma_{\mu, F, Q} := \iint  \mu(a, w) \, F(da) \, Q(dw)$. We thus have $\Omega_{n}^\circ(a_0) =\frac{1}{V}\sum_{v=1}^V \d{P}_{n,v} D_{a_0,n,v}$. If (A3) or (A4) hold, then
\begin{align*}
 P_0 D_{a_0,\infty} &=  \Omega_0(a_0)  + \iint \left[I_{(-\infty, a_0]}(a) - F_0(a_0)  \right] \left[\mu_{\infty}(a,w) - \mu_0(a,w)\right] \left[1 - \frac{g_0(a,w)}{g_{\infty}(a,w)}\right]\, F_0(da) \, Q_0(dw) \\
 &= \Omega_0(a_0)\ . 
 \end{align*}
Thus, we have the first-order expansion $\Omega_{n}^\circ(a_0) -  \Omega_0(a_0) = \d{P}_nD_{a_0, \infty}^* + \frac{1}{V}\sum_{v=1}^V R_{n,a_0,v}$, where 
\[R_{n,a_0,v} :=  (\d{P}_{n,v} - P_0)(D_{a_0,n, v} - D_{a_0, \infty})+\left[P_0 D_{a_0,n, v} - \Omega_0(a_0)\right] + \left( 1 - \frac{VN_v}{n} \right) \d{P}_{n,v} D_{a_0, \infty}^*\ .\]
Straightforward algebra shows that the remainder term $R_{n,a_0,v}$ can be further decomposed as $\sum_{j=1}^5 R_{n,a_0,v,j}$.
\end{proof}

\begin{lemma}\label{lemma:emp_process_neg}
Conditions (A1) and (A2) imply that $\max_v\sup_{a_0 \in\s{A}_0} |\d{G}_{n,v} ( D_{a_0,n,v}^{\circ} - D_{a_0, \infty}^{\circ})| \inprob 0$.
\end{lemma}
\begin{proof}[\bfseries{Proof of Lemma~\ref{lemma:emp_process_neg}}]
We define $\s{F}_{n,v} := \{  D_{a_0,n,v}^{\circ} - D_{a_0, \infty}^{\circ} : a_0 \in \s{A}_0 \}$, so that we can write $\max_v\sup_{a_0 \in\s{A}_0} |\d{G}_{n,v} ( D_{a_0,n,v}^{\circ} - D_{a_0,\infty}^{\circ})|  = \max_v\sup_{f \in\s{F}_{n,v} } |\d{G}_{n,v} f|$. By the tower property, we have
\[ E_0\left[ \sup_{f \in\s{F}_{n,v} } \left|\d{G}_{n,v} f \right| \right] = E_0 \left[ E_0\left(  \sup_{f \in\s{F}_{n,v} } \left|\d{G}_{n,v} f \right| \Bigg| \s{T}_{n,v} \right) \right] \ .\]
The inner expectation is taken with respect to the distribution of the observations with indices in the validation sample $\s{V}_{n,v}$, while the outer expectation is with respect to the observations in the training sample $\s{T}_{n,v}$. By construction, the functions $\mu_{n,v}$ and $g_{n,v}$ depend only upon the observations in the training sample $\s{T}_{n,v}$, so that they are fixed with respect to the inner expectation.  We note that $\sup_{f \in\s{F}_{n,v} } |f(y, a,w)| \leq F_{n,v}(y,a,w)$ for all $y,a,w$, where
\begin{align*}
F_{n,v}(y,a,w)| &= \left[ \left(|y| + K_0\right) K_1^{-1} + K_0\right]\sup_{a_0 \in \s{A}_0} \left| F_{n,v}(a_0) - F_0(a_0)\right| \\
&\qquad+ K_1^{-2}\left( |y| + K_0\right)\left| g_{n,v}(a, w) - g_\infty(a,w) \right| + K_1^{-1}\left|\mu_{n,v}(a,w) - \mu_\infty(a,w)\right| \\
&\qquad+ \int \left| \mu_{n,v}\left(a,\tilde{w}\right) - \mu_\infty\left(a,\tilde{w}\right)\right| \, Q_0\left(d\tilde{w}\right) + \int \left| \mu_{n,v}(u, w) - \mu_\infty(u, w) \right| \, F_0(du) \ .
\end{align*}
We then have by Theorem~2.14.1 of \cite{van1996weak} that 
\[ E_0\left(  \sup_{f \in\s{F}_{n,v} } \left|\d{G}_{n,v} f \right| \Bigg| \s{T}_{n,v} \right)  \leq C \left\{ E_{0}\left[  F_{n,v}(Y, A,W)^2\Bigg| \s{T}_{n,v}\right] \right\}^{1/2}  J\left(1, \s{F}_{n,v}\right) \ , \]
for a constant C not depending on $\s{F}_{n,v}$, where $J$ is the uniform entropy integral as defined in Chapter 2.14 of \cite{van1996weak}. The class $\s{F}_{n,v}$ is a convex combination of the classes (1) $\{I_{(-\infty,a_0]}(a) : a_0 \in \s{A}_0\}$, (2) $\{\int I_{(-\infty,a_0]}(a) F(da) : a_0 \in \s{A}_0\}$ for $F = F_{n,v}$ and $F = F_\infty$, (3) $\{\int I_{(-\infty,a_0]}(a) \mu(a,w)F_0(da) : a_0 \in \s{A}_0\}$ for $\mu = \mu_\infty$ and $\mu = \mu_{n,v}$, and various fixed functions with finite second moments. Class (1) is  well-known to possess polynomial covering numbers. Classes (2) and (3) therefore do as well by Lemma~1 of \cite{westling2020isotonic}. Thus, $\max_v J\left(1, \s{F}_{n,v}\right) = \boundeddet(1)$. Hence, we now have
\begin{align*}
E_0\left[ \sup_{f \in\s{F}_{n,v} } \left|\d{G}_{n,v} f \right| \right] &\leq C' E_0\left( \left\{ E_{0}\left[  F_{n,v}(Y, A,W)^2\Bigg| \s{T}_{n,v}\right] \right\}^{1/2}\right) = C'  E_0 \left[  \|F_{n,v}\|_{P_0,2} \right] \ ,
\end{align*}
The triangle inequality and conditions (A1) and (A2) imply that $\|F_{n,v}\|_{P_0, 2} \inprob 0$ for each $v$, and also that  $\|F_{n,v}\|_{P_0, 2}$ is uniformly bounded for all $n$ and $v$. This implies that $E_0 \left[  \|F_{n,v}\|_{P_0,2} \right] \longrightarrow 0$. Therefore $\sup_{f \in \s{F}_{n,v}} \left|\d{G}_{n,v} f\right| = \sup_{a_0 \in \s{A}_0} \left|\d{G}_n( D_{a_0,n,v}^{\circ} - D_{a_0, \infty}^{\circ})\right| \inprob 0$ for each $v$, which implies that \[\max_v \sup_{a_0 \in \s{A}_0} \left|  \d{G}_n(D_{a_0,n,v}^{\circ} - D_{a_0, \infty}^{\circ})\right| \inprob 0\] since $V = \boundeddet(1)$.
\end{proof}

\begin{lemma}\label{lemma:uprocess}
Condition (A1) implies that 
\[\max_v \sup_{a_0 \in \s{A}_0} \left| \iint \left[ I_{(-\infty, a_0]}(a) - F_{n,v}(a_0)\right]\mu_{n,v}(a, w)(F_{n,v} - F_0)(da) (Q_{n,v} - Q_0)(dw) \right| = \bounded(n^{-1}) \ .\]
\end{lemma}
\begin{proof}[\bfseries{Proof of Lemma~\ref{lemma:uprocess}}]
We have 
\begin{align*}
&\iint \left[ I_{(-\infty, a_0]}(a) - F_{n,v}(a_0)\right]\mu_{n}(a, w)(F_{n,v} - F_0)(da) (Q_{n,v} - Q_0)(dw) \\
&\qquad= \iint I_{(-\infty, a_0]}(a)\mu_{n,v}(a, w)(F_{n,v} - F_0)(da) (Q_{n,v} - Q_0)(dw)  \\
&\qquad\qquad- F_{n,v}(a_0)\iint \mu_{n,v}(a, w)(F_{n,v} - F_0)(da) (Q_{n,v} - Q_0)(dw) \ .
\end{align*}
Controlling these two terms is almost identical, and in fact the second term can be controlled by setting $a_0 = +\infty$. Therefore, we focus only on the first term. 

We write 
\[ \iint I_{(-\infty, a_0]}(a)\mu_{n,v}(a, w)(F_{n,v} - F_0)(da) (Q_{n,v} - Q_0)(dw)  = R_{n, a_0,v, 6} + R_{n, a_0, v,7} + R_{n,a_0, v,8}\] where
\begin{align*}
R_{n,a_0,v,6} &= \frac{1}{2N_v^2} \sum_{\stackrel{i \neq j}{ i, j \in \s{V}_{n,v}}} \gamma_{\mu_{n,v}, a_0}(O_i, O_j) \\
R_{n, a_0,v ,7} &= N_v^{-3/2}\d{G}_{n,v} \omega_{\mu_{n,v}, a_0}\\
R_{n,a_0,v,8} &= N_v^{-1}E_{0} [I_{(-\infty, a_0]}(A) \mu_{n,v}(A, W)]\ ,
\end{align*}
where we have defined $\omega_{\mu, a_0}(y, a, w) := I_{(-\infty, a_0]}(a)\mu(a, w)$ and 
\begin{align*}
\gamma_{\mu, a_0}(o_i, o_j) &:=  I_{(-\infty, a_0]}(a_i) \mu(a_i, w_j) + I_{(-\infty, a_0]}(a_j) \mu(a_j, w_i) \\
&\qquad- \int \left[ I_{(-\infty, a_0]}(a_i) \mu(a_i, w) + I_{(-\infty, a_0]}(a_j)\mu(a_j, w)\right]Q_0(dw) \\
&\qquad- \int_{-\infty}^{a_0} \left[ \mu(a, w_i) + \mu(a, w_j)\right]F_0(da) +  2\int I_{(-\infty, a_0]}(a) \mu(a, w) F_0(da) Q_0(dw)\ .
\end{align*}
For $R_{n,a_0,6}$, we define $\s{G}_{n,v} := \{ \gamma_{\mu_{n,v}, a_0}(O_i, O_j) : a_0 : \s{A}_0\}$ and $S_{n,v}(\gamma) :=  \sum_{\stackrel{i \neq j}{ i, j \in \s{V}_{n,v}}} \gamma(O_i, O_j)$. As in the proof of Lemma~\ref{lemma:emp_process_neg}, we begin by conditioning on $\s{T}_{n,v}$ using the tower property:
\[ E_0 \left[ \sup_{\gamma \in \s{G}_{n,v}} \left| S_{n,v}(\gamma) \right| \right] = E_0 \left\{ E_0 \left[\sup_{\gamma \in \s{G}_{n,v}} \left| S_{n,v}(\gamma) \right| \Bigg| \s{T}_{n,v} \right] \right\} \ .\]
The function $\mu_{n,v}$ is fixed with respect to the inner expectation, so we apply Lemma~2 of \cite{westling2020isotonic} to bound this inner expectation. The class $\s{G}_{n,v}$ is uniformly bounded and satisfies the uniform entropy condition since it is a convex combination of the class $\{ a \mapsto I_{(-\infty, a_0]}(a) : a_0 \in \s{A}_0\}$, various fixed functions, and integrals of the two. Therefore, Lemma~2 of \cite{westling2020isotonic} implies that 
\[ E_0 \left[\sup_{\gamma \in \s{G}_{n,v}} \left| S_{n,v}(\gamma) \right| \Bigg| \s{T}_{n,v} \right] \leq C \left[ N_v (N_v - 1) \right]^{1/2} \]
for some $C < \infty$ not depending on $n$. We thus have that $\sup_{a_0 \in \s{A}_0} \left| R_{n,a_0,v,6}\right| \leq (C /2) N_v^{-1}$, and since $\max_v N_v^{-1} = \boundeddet(n^{-1})$, we then have $\max_v \sup_{a_0 \in \s{A}_0} \left| R_{n,a_0,v,6}\right| = \bounded(n^{-1})$.

For $R_{n,a_0,7}$, since the class of functions $\{\omega_{\mu_{n,v}, a_0} : a_0 \in \s{A}_0\}$ is uniformly bounded almost surely for all $n$ large enough, $\max_v \sup_{a_0 \in \s{A}_0} \left| \d{G}_{n,v}\omega_{\mu_{n,v}, a_0}\right| = \bounded(1)$ by an analogous conditioning argument to that used above. Therefore, $\max_v \sup_{a_0 \in \s{A}_0} \left| R_{n, a_0,v ,7} \right| = \bounded\left(n^{-3/2}\right)$.


Finally, $\max_v\sup_{a_0 \in \s{A}_0}|R_{n,a_0,v,8}| = \bounded(n^{-1})$ since $\max_v |\mu_{n,v}| \leq K_0$ almost surely for all $n$ large enough. This completes the proof
\end{proof}

\begin{proof}[\bfseries{Proof of Theorem~\ref{thm:dr_cons_omega}}]
By Lemma~\ref{lemma:first_order}, we have that 
\[\sup_{a_0 \in \s{A}_0} \left| \Omega_n^\circ(a_0) - \Omega_0(a_0)\right| \leq \sup_{a_0 \in \s{A}_0} \left| \d{P}_{n}D_{a_0, \infty}^* \right| +  \sum_{j=1}^5 \max_v\sup_{a_0 \in \s{A}_0}\left| R_{n, a_0, v,j} \right| \ .\] 
The class $\{D_{a_0, \infty}^* : a_0 \in \d{R}\}$ is $P_0$-Donsker because it is a convex combination of the class $\{I_{(-\infty, a_0]}(a) : a_0 \in \s{A}_0\}$, which is well-known to have polynomial covering numbers, and integrals thereof, which thus also have polynomial covering numbers by  Lemma~1 of \cite{westling2020isotonic}. Since $P_0 D_{a_0, \infty}^* = 0$ for all $a_0$ by (A3), we then have 
\[\sup_{a_0 \in \s{A}_0} \left| \d{P}_nD_{a_0, \infty}^* \right| = n^{-1/2}\sup_{a_0 \in \s{A}_0} \left| \d{G}_nD_{a_0, \infty}^* \right| = \bounded\left(n^{-1/2}\right) \ .\]
Next, we have $\max_v\sup_{a_0 \in \s{A}_0}|R_{n,a_0,v,1}| = n^{-1/2}\max_v \sup_{a_0 \in \s{A}_0} | \d{G}_{n,v}\left( D_{a_0,n,v}^{\circ} - D_{a_0, \infty}^{\circ}\right)| = \fasterthan(n^{-1/2})$ by Lemma~\ref{lemma:emp_process_neg}. Since $\max_v\sup_{a_0 \in \s{A}_0} | F_{n,v}(a_0) - F_0(a_0)| = \bounded(n^{-1/2})$ and \[\max_v \left|\gamma_{\mu_{n,v}, F_0, Q_0} - \gamma_{\mu_{n,v}, F_0, Q_{n,v}}- \gamma_{\mu_0, F_0, Q_0}+ \gamma_{\mu_\infty, F_0, Q_0} \right| = \bounded(1)\ ,\]  $\max_v\sup_{a_0 \in \d{R}} |R_{n,a_0,v,2}| = \bounded(n^{-1/2})$. Additionally, $\max_v\sup_{a_0 \in \s{A}_0} | R_{n, a_0,v, 3}| = \fasterthan(n^{-1/2})$ by Lemma~\ref{lemma:uprocess}.

For $R_{n,a_0,v,4}$ we first have
\begin{align*}
\max_v\sup_{a_0 \in \s{A}_0} | R_{n,a_0,v,4}| &\leq  2K_1^{-2} \max_v \iint \left|\mu_{n,v}(a, w) - \mu_0(a, w)\right| \left| g_{n,v}(a,w)- g_0(a, w)\right| dP_0(a,w) \\
&= 2K_1^{-2} r_n \ .
\end{align*}

Finally, for $R_{n,a_0, v,5}$, since $|N_v - n/V| \leq 1$, $\left| 1 - VN_v / n\right| = O(n^{-1})$, so that $\max_v\sup_{a_0 \in \s{A}_0} |R_{n,a_0, v, 5}| = \fasterthan(n^{-1})$.

We therefore have that 
\[\sup_{a_0 \in \s{A}_0} \left| \Omega_n^\circ(a_0) - \Omega_0(a_0)\right| \leq \bounded\left(n^{-1/2}\right) + 2 r_n = \bounded\left( \max\left\{ n^{-1/2}, r_n\right\} \right) \ . \]
This establishes the first statement in the proof. For the second statement, it suffices to show that $r_n \inprob 0$. For this, we have by (A3) that
\begin{align*}
\frac{K_1^2}{2}r_n &\leq \max_v  \int_{\s{S}_1} \left|\mu_{n,v}- \mu_\infty\right| \left| g_{n,v} - g_0\right| dP_0 + \max_v \int_{\s{S}_2} \left|\mu_{n,v} - \mu_0\right| \left| g_{n,v} - g_\infty\right| dP_0\\
&\qquad +  \int_{\s{S}_3} \left|\mu_{n,v} - \mu_\infty\right| \left| g_{n,v} - g_\infty\right| dP_0 \\
&\leq  \left[P_0(\mu_{n,v} - \mu_\infty)^2 P_0\left( g_{n,v} - g_0\right)^2 \right]^{1/2}   + \left[ P_0(\mu_{n,v} - \mu_0)^2 P_0\left( g_{n,v} - g_\infty\right)^2 \right]^{1/2} \\
&\qquad + \left[P_0(\mu_{n,v} - \mu_\infty)^2 P_0\left(g_{n,v} - g_\infty\right)^2 \right]^{1/2} \ .
\end{align*}
Condition (A2) states that $\max_v P_0(\mu_{n,v} - \mu_\infty)^2 \inprob 0$ and $\max_v P_0(g_{n,v} - g_\infty)^2\inprob 0$, which implies in addition that $\max_v P_0(\mu_{n,v} - \mu_0)^2 = \bounded(1)$, and $\max_v P_0(g_{n,v} - g_0)^2 = \bounded(1)$ by the boundedness condition (A1). Therefore, (A1)--(A3) imply that $r_n \inprob 0$.
\end{proof}

\begin{proof}[\bfseries{Proof of Theorem~\ref{thm:dr_cons_test}}]
The proof proceeds in two steps. First, we show that under the stated conditions, $\|\Omega_n^\circ\|_{F_n, p} \inprob \|\Omega_0 \|_{F_0, p}$ for any $p \in [1, \infty]$. Second, we show  that $T_{n,\alpha,p} / n^{1/2} \inprob 0$. Then we will have that $\|\Omega_n^\circ\|_{F_n, p}  - T_{n,\alpha,p} / n^{1/2} \inprob \|\Omega_0 \|_{F_0, p}$, which is strictly positive by Proposition~\ref{prop:equivalence} since $H_A$ holds. The result follows.

To see that $\|\Omega_n^\circ\|_{F_n, p} \inprob \|\Omega_0 \|_{F_0, p}$, we first write
\begin{align*}
\left|\|\Omega_n^\circ\|_{F_n, p} - \|\Omega_0 \|_{F_0, p}\right|&\leq \left| \|\Omega_n^\circ\|_{F_n, p}  -\|\Omega_0\|_{F_n, p}\right|+ \left|\|\Omega_0\|_{F_n, p} -   \|\Omega_0 \|_{F_0, p}\right|
\end{align*}
The first term is bounded above by $\sup_{a \in \d{R}} \left|  \Omega_n^\circ(a)- \Omega_0(a)\right|$, which by Theorem~\ref{thm:dr_cons_omega} tends to zero in probability under (A1)--(A3). For the second term, for $p < \infty$, $\|\Omega_0\|_{F_n, p}^p \inprob \|\Omega_0 \|_{F_0, p}^p$ by the law of large numbers since $|\Omega_0|^p$ is bounded, which implies by the continuous mapping theorem that $\left|\|\Omega_n^\circ\|_{F_n, p} -   \|\Omega_0 \|_{F_0, p}\right| \inprob 0$. For $p = \infty$, we have $\|\Omega_0\|_{F_n, p} = \sup_{a\in\s{A}_n} |\Omega_0| \leq \|\Omega_0\|_{F_0, p} = \sup_{a\in\s{A}_0} |\Omega_0|$ for all $n$. Let $\varepsilon > 0$, and let $a_0 \in \s{A}_0$ be such that $|\Omega_0(a_0)| > \sup_{a \in \s{A}_0} |\Omega_0(a)| - \varepsilon / 2$. If $a_0$ is a mass point of $F_0$, then $a_0 \in \s{A}_n$ with probability tending to one, so that \[P_0\left(\sup_{a\in\s{A}_n} |\Omega_0| > \sup_{a\in\s{A}_0} |\Omega_0| - \varepsilon / 2\right) \to 1 \ ,\] which implies that $P_0\left(  \left|\|\Omega_0\|_{F_n, \infty} -   \|\Omega_0 \|_{F_0, \infty}\right| < \varepsilon\right) \to 1$. If $a_0$ is not a mass point of $F_0$, then $\Omega_0$ must be continuous at $a_0$, so that there exists a $\delta > 0$ such that $|\Omega_0(a) - \Omega_0(a_0)| < \varepsilon / 2$ for all $a$ such that $|a - a_0| < \delta$. Then $|\Omega_0(a)| > \|\Omega_0\|_{F_0, \infty} - \varepsilon$ for all such $a$. Since $P_0(\s{A}_n \cap (a_0 - \delta, a_0 + \delta) = \emptyset) \to 0$, we then have $P_0\left(  \left|\|\Omega_0\|_{F_n, \infty} -   \|\Omega_0 \|_{F_0, \infty}\right| < \varepsilon\right) \to 1$. In either case, since $\varepsilon$ was arbitrary, we have that $\left|\|\Omega_n^\circ\|_{F_n, p} -   \|\Omega_0 \|_{F_0, p}\right| \inprob 0$.

We have now shown that $\|\Omega_n^\circ\|_{F_n, p} \inprob \|\Omega_0 \|_{F_0, p}$, and it remains to show that $T_{n,\alpha,p} / n^{1/2} \inprob 0$. We recall that $T_{n,\alpha,p}$ is defined as $T_{n,\alpha,p} := \inf\left\{t : P_0\left( \|Z_n\|_{F_n,p} \leq t \mid O_1, \dotsc, O_n\right)  \geq 1 - \alpha\right\}$, where $Z_n$ is a mean-zero Gaussian process on $\s{A}_n := \{A_1, \dotsc, A_n\}$ with covariance  given by 
\[ \Sigma_n(s,t) := E_0 \left[ Z_n(s) Z_n(t) \mid O_1, \dotsc, O_n\right]  =  \frac{1}{V}\sum_{v=1}^V \d{P}_{n,v} \left( D_{s, n, v}^* D_{t, n, v}^* \right) \ .\]
 (The dependence on $O_1, \dotsc, O_n$ in the probability is due to $\Sigma_n$ depending on $O_1, \dotsc, O_n$.)  Therefore, $T_{n,\alpha, p} / n^{1/2} > \varepsilon$ implies that $P_0\left( \|Z_n\|_{F_n,p} / n^{1/2} > \varepsilon \mid O_1, \dotsc, O_n\right) \geq \alpha$, which further implies that $P_0\left( \sup_{a \in \s{A}_n} \left|Z_n(a) / n^{1/2} \right|  > \varepsilon \mid O_1, \dotsc, O_n\right) \geq \alpha$ since $\sup_{a \in \s{A}_n} \left|Z_n(a) \right| \geq  \|Z_n\|_{F_n,p}$ for all $p \in [1, \infty]$.  By Markov's inequality, we then have
\[P_0\left( T_{n,\alpha, p} / n^{1/2} > \varepsilon\right) \leq P_0\left( E_0\left[ \sup_{a \in \s{A}_n} \left|Z_n(a) / n^{1/2} \right| \Bigg| O_1, \dotsc, O_n\right] \geq \varepsilon\alpha\right) \ . \]
We define $\rho_n(s,t) := \left[ \Sigma_n(s,s) - 2\Sigma_n(s,t) + \Sigma_n(t,t)\right]^{1/2} / n^{1/2}$. Then, since $Z_n / n^{1/2}$ is a Gaussian process with covariance $\Sigma_n / n$, it is sub-Gaussian with respect to its intrinsic semimetric $\rho_n$, so that 
\begin{align*}
E_0\left[ \sup_{a \in \s{A}_n} \left|Z_n(a) / n^{1/2} \right| \Bigg| O_1, \dotsc, O_n\right] & \leq C\left\{\Sigma_n(a_0, a_0)^{1/2} / n^{1/2} + \int_0^\infty \left[ \log N(\varepsilon, \s{A}_n, \rho_n) \right]^{1/2} \, d\varepsilon\right\}
 \end{align*}
for any $a_0 \in \s{A}_n$ by Corollay~2.2.8 of \cite{van1996weak}. Here, $N(\varepsilon, \s{A}_n, \rho_n)$ is the minimal number of $\rho_n$ balls of radius $\varepsilon$ required to cover $\s{A}_n$. We note that for $\varepsilon \geq \left(\|\Sigma_n\|_{\infty} / n\right)^{1/2}$, $N(\varepsilon, \s{A}_n, \rho_n) = 1$, since it only takes one $\rho_n$ ball of radius $\left(\|\Sigma_n\|_\infty / n\right)^{1/2}$ to cover $\s{A}_n$. For $\varepsilon \leq \left(\|\Sigma_n\|_\infty / n\right)^{1/2}$, we have the trivial inequality $N(\varepsilon, \s{A}_n, \rho_n) \leq n$, since $|\s{A}_n| \leq n$. Thus, we have almost surely for all $n$ large enough that
\begin{align*}
  \int_0^\infty \left[ \log N(\varepsilon, \s{A}_n, \rho_n) \right]^{1/2} \, d\varepsilon &\leq  \int_0^{ \left(M / n\right)^{1/2}}\left[ \log n \right]^{1/2} \, d\varepsilon =  \left(M  n^{-1} \log n \right)^{1/2} \ .
\end{align*}
Therefore,
\begin{align*}
&  P_0\left( E_0\left[ \sup_{a \in \s{A}_n} \left|Z_n(a) / n^{1/2} \right| \Bigg| O_1, \dotsc, O_n\right] \geq \varepsilon\alpha\right) \\
&\qquad\leq P_0\left( C\Sigma_n(a_0, a_0)^{1/2} / n^{1/2} +  C\left(\|\Sigma_n\|_\infty n^{-1} \log n \right)^{1/2} \geq \varepsilon \alpha \right) \ .
\end{align*}
It is straightforward to see that condition (A1) implies that $\sup_{s, t \in \s{A}_0}|\Sigma_n(s,t)| = \bounded(1)$, so that the last probability tends to zero. for any $\epsilon, \alpha > 0$. Therefore, $P_0\left( T_{n,\alpha, p} / n^{1/2} > \varepsilon\right) \longrightarrow 0$ for any $\varepsilon > 0$, so that $ T_{n,\alpha, p} / n^{1/2} \inprob 0$. This completes the proof.
%
\end{proof}

\begin{proof}[\bfseries{Proof of Theorem~\ref{thm:weak_conv_omega}}]
As in the proof of Theorem~\ref{thm:dr_cons_omega}, $\max_v\sup_{a_0 \in \s{A}_0}|R_{n,a_0,v,1}| = \fasterthan(n^{-1/2})$ by Lemma~\ref{lemma:emp_process_neg}, $\max_v\sup_{a_0 \in \s{A}_0} | R_{n, a_0,v, 3}| = \fasterthan(n^{-1/2})$ by Lemma~\ref{lemma:uprocess}, $\max_v\sup_{a_0 \in \s{A}_0} | R_{n,a_0,v,4}|  =  \bounded(r_n) = \fasterthan(n^{-1/2})$ by assumption, and $\max_v\sup_{a_0 \in \s{A}_0} | R_{n,a_0,v,5}| = \fasterthan(n^{-1})$. For $R_{n,a_0,v,2}$, since $\mu_\infty = \mu_0$, we have
\[ R_{n,a_0,v,2} = \left(\gamma_{\mu_{n,v}, F_0, Q_0} - \gamma_{\mu_{n,v}, F_0, Q_{n,v}}\right) \left[ F_{n,v}(a_0) - F_0(a_0)\right] = \left(N_v^{-1/2} \d{G}_{n,v} \eta_{\mu_{n,v}, F_0}\right) \bounded(n^{-1/2}) \ ,\]
where we define $\eta_{\mu, F}(w) := \int \mu(a,w) \, F(da)$. Since $\eta_{\mu_{n,v}, F_0}$ is a fixed function relative to $\s{V}_{n,v}$, $\max_v \d{G}_{n,v} \eta_{\mu_{n,v}, F_0} = \bounded(1)$, so that $\max_v \sup_{a_0 \in \s{A}_0} |R_{n,a_0,v,2}| = \bounded(n^{-1})$.

We now have $ \Omega_n^\circ(a) - \Omega_0(a) = \d{P}_n D_{a, 0}^* + R_{n,a},$
where $\sup_{a \in \s{A}_0} |R_{n,a}| = \fasterthan(n^{-1/2})$. Since $\{ D_{a,0}^* : a \in \s{A}_0\}$ is a $P_0$-Donsker class, the result follows.
\end{proof}
Before proving Theorem~\ref{thm:weak_conv_omega}, we introduce several additional Lemmas. First, we demonstrate that $\Sigma_n$ is a uniformly consistent estimator of the limiting covariance $\Sigma_0$.
\begin{lemma}\label{lemma:covar_cons}
If the conditions of Theorem~\ref{thm:test_size} hold, then $E_0\left[\sup_{(s,t) \in \s{A}_0^2} | \Sigma_n(s,t) - \Sigma_0(s,t)| \right]\longrightarrow 0$.
\end{lemma}
\begin{proof}[\bfseries{Proof of Lemma~\ref{lemma:covar_cons}}]
We recall that $\Sigma_n(s,t) := \frac{1}{V}\sum_{v=1}^V \d{P}_{n,v} [D_{s,n,v}^* D_{t,n,v}^*]$ and $\Sigma_0(s,t) := P_0[ D_{s,0}^* D_{t,0}^*]$. We can thus write
\begin{align*}
\nonumber\Sigma_n(s,t)  - \Sigma_0(s,t) &= \frac{1}{V} \sum_{v=1}^V \left[ (\d{P}_{n,v} - P_0)(D_{s,n,v}^* D_{t,n,v}^*) + P_0(D_{s,n,v}^* D_{t,n,v}^* - D_{s,0}^* D_{t,0}^*) \right]
\end{align*}
Therefore,
\begin{align}
E_0\left[\sup_{(s,t) \in \s{A}_0^2} | \Sigma_n(s,t) - \Sigma_0(s,t)| \right] &\leq \max_v N_v^{-1/2} E_0 \left[\sup_{(s,t) \in \s{A}_0^2}\left| \d{G}_{n,v}(D_{s,n,v}^* D_{t,n,v}^*)\right| \right]  \nonumber\\
&\qquad + \max_v E_0 \left[\sup_{(s,t) \in \s{A}_0^2}\left| P_0(D_{s,n,v}^* D_{t,n,v}^* - D_{s,0}^* D_{t,0}^*)  \right| \right] \label{eq:decomp_sigma}\ .
\end{align}
For the first term, a conditioning argument analogous to that in the proof of  Lemma~\ref{lemma:emp_process_neg} in conjunction with Theorem~2.14.1 of \cite{van1996weak} implies that \[\max_v E_0 \left[\sup_{(s,t) \in \s{A}_0^2}\left| \d{G}_{n,v}(D_{s,n,v}^* D_{t,n,v}^*)\right| \right] = O(1) \ ,\]since $\{ D_{s,n,v}^* D_{t,n,v}^* : (s,t) \in \s{A}_0^2\}$ satisfies a suitable entropy bound conditional upon the nuisance function estimators by permanence properties of entropy bounds. Therefore, the first term is $\boundeddet(n^{-1/2})$, and in particular is $\fasterthandet(1)$.

For the second term, we note that
\begin{align*}
 P_0 \left| D_{s,n,v}^* D_{t,n,v}^* - D_{s,0}^* D_{t,0}^*\right|  &\leq  P_0 \left|\left(D_{s,n,v}^* - D_{s,0}^*\right) D_{t,0}^* \right| +P_0 \left|\left(D_{t,n,v}^* - D_{t,0}^*\right) D_{s,n,v}^* \right|  \\
&\leq \left\{ P_0 \left(  D_{s,n,v}^* - D_{s,0}^*\right)^2P_0 \left(D_{t,0}^*\right)^2 \right\}^{1/2} \\
&\qquad+\left\{P_0  \left(D_{t,n,v}^* - D_{t,0}^*\right)^2   P_0 \left( D_{s,n,v}^* \right)^2\right\}^{1/2}  \ .
\end{align*}
Since $P_0 (D_{s,n,v}^*)^2$ and $P_0 (D_{t,0}^*)^2$ are uniformly bounded for all $n$ large enough by condition (A1), the preceding display is bounded up to a constant by $P_0 \left(F_{n,v}^2\right)$ for $F_{n,v}$ defined in the proof of Lemma~\ref{lemma:emp_process_neg}. This tends to zero in expectation uniformly over $v$ by an argument analogous to that proof of  Lemma~\ref{lemma:emp_process_neg} and the assumption that $r_n = \fasterthan(n^{-1/2})$. This implies that the second term in \eqref{eq:decomp_sigma} tends to zero.
\end{proof}

Given $O_1, \dotsc, O_n$, let $Z_n$ be distributed according to a mean-zero Gaussian process with covariance $\Sigma_n$ as defined in the main text. The next lemma shows that $Z_n$ converges weakly in $\ell^\infty(\s{A}_0)$ to the limiting Gaussian process $Z_0$ with covariance $\Sigma_0$.
\begin{lemma}\label{lemma:Zn_conv}
If the conditions of Theorem~\ref{thm:test_size} hold, then $Z_n$ converges weakly in $\ell^\infty(\s{A}_0)$ to the limiting Gaussian process $Z_0$.
Furthermore, $\|Z_n\|_{F_n,p} - \|Z_n\|_{F_0, p} \inprob 0$, so that $\|Z_n\|_{F_n,p} \indist \|Z_0\|_{F_0, p}$ for any $p \in [1, \infty]$. 
\end{lemma}
\begin{proof}[\bfseries{Proof of Lemma~\ref{lemma:Zn_conv}}]
We first demonstrate that the finite-dimensional marginals of $Z_n$ converge in distribution to the finite-dimensional marginals of $Z_{0}$. We let $\Sigma_{n,a}$ be the covariance matrix of $(Z_n(a_1), \dotsc, Z_n(a_m))$ and $\Sigma_{0,a}$ be the covariance matrix of $Z_{0,a} = (Z_0(a_1), \dotsc, Z_0(a_m))$. We then have since $Z_n$ is a mean-zero Gaussian process conditional on $O_1, \dotsc, O_n$ and $Z_0$ is a mean-zero Gaussian process that
\begin{align*}
\left|E_0\left[ \exp\{ i t^T Z_{n,a} \} \right] - E_0\left[\exp\{ i t^T \d{G}_{0,a} \} \right]\right| &= \left|E_0\left\{ E\left[ \exp\{ i t^T Z_{n,a} \} \mid O_1, \dotsc, O_n\right]  \right\}- E_0\left[\exp\{ i t^T \d{G}_{0,a} \} \right]\right|\\
& =\left|E_0\left[ \exp\left\{ -\tfrac{1}{2} t^T \Sigma_{n,a} t\right\} -\exp\left\{ -\tfrac{1}{2} t^T \Sigma_{0,a} t\right\} \right]\right| \\
&\leq E_0 \left| \tfrac{1}{2} t^T  \left(\Sigma_{n,a}- \Sigma_{0,a}\right) t\right| \\
&\leq E_0\left[ \sup_{s,t} |\Sigma_{n}(s,t)- \Sigma_{0}(s,t)| \right]  \sum_{i,j} |t_i t_j| \ ,
\end{align*}
which tends to zero for every $t$ by Lemma~\ref{lemma:covar_cons}. Therefore, 
\[(Z_n(a_1), \dotsc, Z_n(a_m)) \indist  (Z_0(a_1), \dotsc, Z_0(a_m))\]
 for any $(a_1, \dotsc, a_m) \in \s{A}_0^m$ and $m \in \{1, 2, \dotsc\}$.

In order to show that $Z_n$ converges weakly in $\ell^\infty(\s{A}_0)$ to the limiting Gaussian process $Z_0$, we need also to demonstrate asymptotic uniform mean-square equicontinuity, meaning that for all $\varepsilon$ and $\eta > 0$, there exists $\delta > 0$ such that 
\[ P_0 \left( \sup_{d_0(s,t) < \delta} |Z_n(s) - Z_n(t)| > \varepsilon \right) < \eta \ ,\]
where $d_0(s,t) := |F_0(s) - F_0(t)|^{1/2}$. We define $d_n(s, t):= |F_n(s) - F_n(t)|^{1/2}$. Then $\sup_{(s,t) \in \s{A}_0^2} |d_n(s,t) - d_0(s,t)| \inprob 0$. We note that since $Z_n$ is a Gaussian process conditional on $O_1,\dotsc,O_n$ with covariance $\Sigma_n$, it is sub-Gaussian with respect to the semi-metric $\rho_n$ given by $\rho_n(s,t) := [\Sigma_n(s,s) + \Sigma_n(t,t) - 2\Sigma_n(s,t)]^{1/2}$. Furthermore, it is straightforward to verify that condition (A1) implies that $\rho_n(s,t) \leq Cd_n(s,t)$ for all $(s,t) \in \s{A}_0^2$ and some $C < \infty$ not depending on $n$, so that $Z_n$ is sub-Gaussian with respect to $d_n$ as well. Therefore, by Corollary~2.2.8 of \cite{van1996weak}, 
\[ E\left[ \sup_{d_n(s,t) < \delta} |Z_n(s) - Z_n(t)| \mid O_1, \dotsc O_n \right] \leq C' \int_0^\delta \left[ \log N(x, \s{A}_0, d_n)\right]^{1/2} \, dx \]
for every $\delta > 0$ and some $C < \infty$ not depending on $n$ or $\delta$, where, as before, $N(x, \s{A}_0, d)$ is the minimal number of $d$-balls of radius $x$ required to cover $\s{A}_0$. For $x < n^{-1/2}$, $N(x, \s{A}_0, d_n) \leq n$, and $N(x, \s{A}_0, d_n) \leq x^{-2}$ otherwise, so that 
\[ E\left[ \sup_{d_n(s,t) < \delta} |Z_n(s) - Z_n(t)| \mid O_1, \dotsc O_n \right] \leq C''\left[ \left(\log n\right)^{1/2} n^{-1/2} + h(\delta)\right]\ , \]
where $h(x) = x \left[ \log (1/x)\right]^{1/2}$, which tends to zero as $x \to 0$. Thus, for any $\alpha > 0$ we have
\begin{align*}
&P_0\left( \sup_{d_0(s,t) < \delta} |Z_n(s) - Z_n(t)| > \varepsilon\right) = E_0\left[ P_0\left( \sup_{d_0(s,t) < \delta} |Z_n(s) - Z_n(t)| > \varepsilon \mid O_1, \dotsc O_n\right) \right]\\
&\qquad\qquad\leq E_0\left[ P_0\left( \sup_{d_0(s,t) < \delta} |Z_n(s) - Z_n(t)| > \varepsilon \mid \| d_n - d_0 \|_\infty < \alpha, O_1, \dotsc O_n\right)\right]\\
&\qquad\qquad \qquad+ P_0\left(\| d_n - d_0 \|_\infty \geq \alpha \right)  \\
&\qquad\qquad\leq E_0\left[ P_0\left( \sup_{d_n(s,t) < \delta + \alpha} |Z_n(s) - Z_n(t)| > \varepsilon \mid \| d_n - d_0 \|_\infty < \alpha, O_1, \dotsc O_n\right)\right] \\
&\qquad\qquad \qquad+ P_0\left(\| d_n - d_0 \|_\infty \geq \alpha \right)  \\
&\qquad\qquad\leq \varepsilon^{-1} C''\left[ \left(\log n\right)^{1/2}\min\{\delta + \alpha, n^{-1/2}\} + h(\delta + \alpha) \right] + P_0\left(\| d_n - d_0 \|_\infty \geq \alpha \right) \ .
\end{align*}
We can choose $\delta$ and $\alpha$ such that $C'' h(\delta + \alpha) /\varepsilon < \eta / 3$. For any such fixed $\delta$ and $\alpha$, $n^{-1/2} < \delta + \alpha$ and $\varepsilon^{-1} C''  \left( n^{-1} \log n\right)^{1/2}  < \eta/ 3$ for all $n$ large enough. Finally, for any $\alpha > 0$, $P_0\left(\| d_n - d_0 \|_\infty \geq \alpha \right) < \eta /3$ for all $n$ large enough since $\| d_n - d_0 \|_\infty \inprob 0$. We thus have that the limit superior as $n \to \infty$ of the preceding display is smaller than $\eta$.

For the claim that $\|Z_n\|_{F_n,p} - \|Z_n\|_{F_0, p} \inprob 0$, we first note that $Z_n(s) = Z_n(t)$ almost surely for any $s,t$ such that $F_n(s) = F_n(t)$ since $\rho_n(s,t)^2 = E \left\{ \left[Z_n(s) - Z_n(t)\right]^2\right\} \leq  C |F_n(s) - F_n(t)|$. Therefore, $Z_n$ is almost surely a right-continuous step function with steps at $A_1, \dotsc, A_n$, so that $\|Z_n\|_{F_n, \infty} = \|Z_n\|_{F_0, \infty}$ almost surely. For the case that  $p \in [1, \infty)$, we let $\varepsilon > 0$. Then, for any $\delta, \gamma > 0$ we have
\begin{align*}
P_0\left(\left| \|Z_n\|_{F_n,p}-\|Z_n\|_{F_0,p} \right| > \varepsilon  \right) &\leq P_0\left(\left| \|Z_n\|_{F_n,p}-\|Z_n\|_{F_0,p} \right| > \varepsilon \, \Bigg| \,  \sup_{d_0(s,t) < \delta} |Z_n(s) - Z_n(t)| \leq  \gamma \right) \\
&\qquad \qquad + P_0\left(  \sup_{d_0(s,t) < \delta} |Z_n(s) - Z_n(t)| > \gamma \right)  \ .
\end{align*}
The second term tends to zero by the above. For the first term, we let $\s{A}_1^+, \dotsc \s{A}_m^+$ be intervals covering $\s{A}_0$ such that $\s{A}_0 \cap \s{A}_j^+ \neq \emptyset$ and such that $\max_{1 \leq j\leq m}F_0\left(\s{A}_j^+\right) \leq \delta^2$. This can be done with $m \leq 2 \delta^{-2}$ intervals. We let $a_j \in \s{A}_j^+ \cap \s{A}_0$ for each $j$. We then define $Z_n^+$ as the stochastic process on $\s{A}_0$ such that $Z_n^+(a) = Z_n(a_j)$ for all $a \in \s{A}_j^+$ for each $j \in \{1, \dotsc, m\}$. Given that $\sup_{d_0(s,t) < \delta} |Z_n(s) - Z_n(t)| \leq \gamma$, we then have
\begin{align*}
\left| \|Z_n\|_{F_n, p} - \|Z_n^+\|_{F_n, p} \right| &\leq \|Z_n - Z_n^+\|_{F_n, p} = \left[ \sum_{j=1}^m \int_{\s{A}_j^+} \left| Z_n(a) - Z_n^+(a_j)\right|^p \, dF_n(a) \right]^{1/p} \leq \gamma 
\end{align*}
and
\begin{align*}
\left| \|Z_n\|_{F_0, p} - \|Z_n^+\|_{F_0, p} \right| &\leq \|Z_n - Z_n^+\|_{F_0, p} = \left[ \sum_{j=1}^m \int_{\s{A}_j^+} \left| Z_n(a) - Z_n^+(a_j)\right|^p \, dF_0(a) \right]^{1/p} \leq \gamma \ .
\end{align*}
Therefore, if $\sup_{d_0(s,t) < \delta} |Z_n(s) - Z_n(t)| \leq \gamma$, then
\begin{align*}
 \left| \|Z_n\|_{F_n,p}-\|Z_n\|_{F_0,p} \right| &\leq 2\gamma + \left| \|Z_n^+\|_{F_n, p} - \|Z_n^+\|_{F_0, p} \right|  \leq 2 \gamma + \left| \sum_{j=1}^m \int_{\s{A}_j^+} \left|Z_n^+(a)\right|^p (F_n - F_0)(da)  \right|^{1/p} \ .
  \end{align*}
  Now, since $Z_n^+(a) = Z_n(a_j)$ for all $a \in \s{A}_j$,  
  \[\int_{\s{A}_j^+} \left|Z_n^+(a)\right|^p (F_n - F_0)(da) = |Z_n(a_j)|^p \left[ F_n\left(\s{A}_j^+\right) - F_0\left(\s{A}_j^+\right)\right] \ .\]
   Furthermore, $\left|F_n\left(\s{A}_j\right) - F_0\left(\s{A}_j^+\right) \right| \leq 2 \| F_n-F_0\|_\infty$ since $\s{A}_j^+$ is an interval. Therefore, 
  \[ \left|  \sum_{j=1}^m \int_{\s{A}_j^+} \left|Z_n^+(a)\right|^p (F_n - F_0)(da) \right|^{1/p} \leq \left( 2 m\| Z_n\|_{\infty}^p \| F_n-F_0\|_\infty \right)^{1/p}  \leq \left(4 \delta^{-2}\right)^{1/p}  \| Z_n\|_{\infty} \| F_n-F_0\|_\infty^{1/p} \ , \]
  which implies that
\[ \left| \|Z_n\|_{F_n,p}-\|Z_n\|_{F_0,p} \right|   \leq  2 \gamma + 4\delta^{-2}\| Z_n\|_{\infty} \| F_n-F_0\|_\infty^{1/p} \ . \]
Hence, setting $\gamma = \varepsilon / 4$ and $\delta = \varepsilon$, we have that 
\begin{align*}
&P_0\left(\left| \|Z_n\|_{F_n,p}-\|Z_n\|_{F_0,p} \right| > \varepsilon \, \Bigg| \,  \sup_{d_0(s,t) < \delta} |Z_n(s) - Z_n(t)| \leq  \gamma\right) \\
&\qquad \leq P_0\left(\| Z_n\|_{\infty} \| F_n-F_0\|_\infty^{1/p} > (\varepsilon/2)^{1+2/p}\, \Bigg| \,  \sup_{d_0(s,t) < \delta} |Z_n(s) - Z_n(t)| \leq  \gamma\right)\ .
\end{align*}
Since $\|Z_n \|_{\infty} = \bounded(1)$ and $\| F_n-F_0\|_\infty^{1/p} = \fasterthan(1)$, this tends to zero for any $\varepsilon > 0$ and $p \in [1, \infty)$, which completes the proof.

Since $\|\cdot\|_{F_0,p}$ is a continuous mapping on $\ell^\infty(\s{A}_0)$, $\|Z_n\|_{F_0,p} \indist \|Z_0\|_{F_0, p}$ for any $p \in [1, \infty]$. Therefore, $\|Z_n\|_{F_n,p} \indist \|Z_0\|_{F_0, p}$ as well.

\end{proof}

\begin{lemma}\label{lemma:Omegan_norm}
If the conditions of Theorem~\ref{thm:test_size} hold, then $\|\Omega_n^\circ\|_{F_n,p} - \|\Omega_n^\circ\|_{F_0,p}  = \fasterthan\left( n^{-1/2}\right)$ for any $p \in [1, \infty]$.
\end{lemma}
\begin{proof}[\bfseries{Proof of Lemma~\ref{lemma:Omegan_norm}}]
We first note that $\Omega_n^\circ$ is a right-continuous step function with steps at $A_1, \dotsc, A_n$, and that $\Omega_n^\circ(a) = 0$ for $a < \min_i A_i$. Therefore, since each $A_i \in \s{A}_0$, $\|\Omega_n^\circ\|_{F_n, \infty} =  \|\Omega_n^\circ\|_{F_0, \infty}$. For $p < \infty$, we have
\begin{align*}
n^{1/2}\left| \|\Omega_n^\circ\|_{F_n, p}  - \|\Omega_n^\circ\|_{F_0, p}\right| &= n^{1/2}\left| \left( \int |\Omega_n^\circ|^p \, dF_n \right)^{1/p} - \left( \int |\Omega_n^\circ|^p \, dF_0 \right)^{1/p} \right| \\
&\leq  \left| n^{p/2} \int |\Omega_n^\circ|^p \, d(F_n-F_0)  \right|^{1/p} =   \left| n^{\frac{p-1}{2}}\d{G}_n  |\Omega_n^\circ|^p  \right|^{1/p} \ .
\end{align*}
Therefore, if we can demonstrate that $ |\Omega_n^\circ|^p$ in contained in a class of functions $\s{G}_{n,p}$ such that $E_0 \left[ \sup_{g \in \s{G}_{n,p}} | \d{G}_n g| \right] = \fasterthandet\left( n^{-\frac{p-1}{2}}\right)$, then we will have that $n^{1/2}\left| \|\Omega_n^\circ\|_{F_n, p}  - \|\Omega_n^\circ\|_{F_0, p}\right| \inprob 0$. In order to show this, we will need boundedness which only holds in probability, but not almost surely for all $n$ large enough. Thus, for any $\eta > 0$, we write
\begin{align*}
P_0 \left( \left| n^{\frac{p-1}{2}} \d{G}_n  |\Omega_n^\circ|^p  \right| > \eta\right) &\leq  P_0 \left( \left| n^{\frac{p-1}{2}}\d{G}_n |\Omega_n^\circ|^p \right| > \eta \Bigg| \left\| \Omega_n^\circ \right\|_\infty \leq n^{-\alpha}, \frac{1}{n} \sum_{i=1}^n |Y_i| \leq E_0[|Y|] + 1\right) \\
&\qquad + P_0\left( \left\| \Omega_n^\circ \right\|_\infty > n^{-\alpha}\right) + P_0 \left( \frac{1}{n} \sum_{i=1}^n |Y_i| > E_0[|Y|] + 1 \right) \ .
\end{align*}
The final probability on the right tends to zero since $\frac{1}{n} \sum_{i=1}^n |Y_i| \inprob E_0[ |Y| ]$. The second probability on the right side tends to zero for any fixed $\alpha \in [0,1/2)$ since $\left\|n^{\alpha}\Omega_n^\circ \right\|_\infty = n^{\alpha - 1/2} \left\| n^{1/2} \Omega_n^\circ \right\|_\infty = \bounded\left(n^{\alpha - 1/2}\right)$. Now we can focus on the first probability. We note that, with some rearranging, we can write $\Omega_n^\circ(a_0) = \sum_{i=1}^n \omega_{n,i} I_{[A_i, \infty)}(a_0)$, where
\begin{align*}
\omega_{n,i} &:= \frac{1}{V N_{v_i}} \frac{Y_i - \mu_{n,v_i}(A_i, W_i)}{ g_{n,v_i}(A_i ,W_i) } - \sum_{O_j \in \s{T}_{n, v_i}} \frac{1}{V N_{v_j} (n - N_{v_j})} \frac{Y_j - \mu_{n, v_j}(A_j, W_j)}{g_{n,v_j}(A_j, W_j)} \\
&\qquad + \frac{1}{V N_{v_i}^2} \sum_{j \in \s{V}_{n,v_i}} \mu_{n, v_i}(A_i, W_j)  - \sum_{O_k \in \s{T}_{n,v_i}} \sum_{j \in \s{V}_{n,v_k}} \frac{\mu_{n, v_k}(A_k, W_j) }{V N_{v_k}^2 (n - N_{v_k})} \ ,
\end{align*}
where $v_i$ is the unique element of $\{1, \dotsc, V\}$ such that $i \in \s{V}_{n,v_i}$. Using the boundedness condition (A1) and the fact that $\frac{1}{V N_v} \leq \frac{2}{n}$ for each $v$,  it is straightforward to see that
\[ | \omega_{n,i}| \leq \frac{2}{n} \left[ K_1^{-1} |Y_i| + \left(K_1^{-1} + 1\right) K_0  + \sum_{O_j \in \s{T}_{n,v_i}}\frac{ K_1^{-1} |Y_j| + K_0}{n - N_{v_j}}  \right]\ , \]
which, if $\frac{1}{n} \sum_{i=1}^n |Y_i| \leq E_0[|Y|] + 1$, implies that
\[ \sum_{i=1}^n  | \omega_{n,i}| \leq 2\left[ 2K_1^{-1} \left( E_0[|Y|] + 1\right) + \left(K_1^{-1} + 2\right) K_0 \right] =: C \ .\]
We then have that $C^{-1} \Omega_n^\circ(a_0) = \sum_{i=1}^n \lambda_{n,i} I_{[A_i, \infty)}(a_0)$, where  $\lambda_{n,i} := \omega_{n, i} / C$ satisfy $\sum_{i=1}^n |\lambda_{n,i}| \leq 1$. Thus, if $\frac{1}{n} \sum_{i=1}^n |Y_i| \leq E_0[|Y|] + 1$, then $C^{-1} \Omega_n^\circ$ is contained in the symmetric convex hull $\s{F}$ of the class $\{ x \mapsto I_{[a, \infty)}(x) : a \in \d{R}\}$. Since this latter class has VC index 2, by Theorem~2.6.9 of \cite{van1996weak}, $\s{F}$ satisfyies $\log N(\varepsilon, \s{F}, L_2(Q)) \leq D \varepsilon^{-1}$ for all probability measures $Q$ and for a constant $D$ not depending on $\varepsilon$ or $Q$. Thus, if both $\frac{1}{n} \sum_{i=1}^n |Y_i| \leq E_0[|Y|] + 1$ and $\left\| \Omega_n^\circ \right\|_\infty \leq n^{-\alpha}$, we have that $\Omega_n^\circ$ is contained in the class $\s{F}_{n} := \{ Cf : f \in \s{F}, \|Cf\|_\infty \leq n^{-\alpha} \}$ with envelope function $R_n(x) = n^{-\alpha}$. Since $\s{F}_n \subseteq C\s{F}$, $\s{F}_n$ satisfies the same entropy bound as $\s{F}$ up to the constant $D$. Hence $|\Omega_n^\circ|^p$ is contained in $\left|\s{F}_{n}\right|^p := \{ |f|^p : f \in \s{F}_n\}$ with envelope $R_{n}^p = n^{-p\alpha}$. Since the function $x \mapsto |x|^p$ is convex for $p \geq 1$, we have $\left| |f|^p - |g|^p\right| \leq |f - g| p R_{n}^{p-1}$ for $f,g \in \s{F}_n$. By Theorem~2.10.20 of \cite{van1996weak} (or Lemma~5.1 of \cite{vaart2006survival}), we then have 
\[\sup_Q \log N\left(\varepsilon \|p R_{n}^{p} \|_{Q,2}, \left| \s{F}_n\right|^p, L_2(Q)\right) \leq \sup_Q  \log N\left(\varepsilon \| R_n \|_{Q,2},\s{F}_n , L_2(Q)\right)  \leq D \left(\varepsilon n^{-\alpha}\right)^{-1}\ . \]
Theorem~2.14.1 of \cite{van1996weak} then implies that for $n$ large enough
\begin{align*}
E_0 \left| \d{G}_n \left| \Omega_n^\circ\right|^p \right| &\leq E_0\left[ \sup_{g \in |\s{F}_n|^p} \left| \d{G}_n g\right| \right] \leq C' \| R_n^p \|_{P,2} \int_0^1 \left[ 1 + \sup_Q \log N\left(\varepsilon \|R_n^p \|_{Q,2}, \left| \s{F}_n\right|^p, L_2(Q)\right)\right]^{1/2} \, d\varepsilon \\
&\leq C' n^{-p\alpha} \int_0^1 \left[ 1 + D \left(\varepsilon n^{-\alpha}/p\right)^{-1}\right]^{1/2} \, d\varepsilon  = C' p n^{(1-p)\alpha} \int_0^{n^{-\alpha}/p} \left[ 1 + D / \varepsilon\right]^{1/2} \, d\varepsilon  \\
&\leq C' p n^{(1-p)\alpha} \int_0^{n^{-\alpha}/p} \left[ 2D / \varepsilon\right]^{1/2} \, d\varepsilon\\
 &= C'' n^{(1-p)\alpha} n^{-\alpha / 2} = \boundeddet\left(n^{(1/2 - p)\alpha}\right) \ .
\end{align*}
Thus, we have $n^{\frac{p-1}{2}}\d{G}_n  |\Omega_n^\circ|^p = \bounded\left(n^{(p-1)/2+ (1/2 - p)\alpha}\right)$. Since $(p-1)/2 + (1/2 - p)\alpha < 0$ for any $\alpha > \frac{p-1}{2p-1}$ and $\frac{p-1}{2p-1} < \frac{1}{2}$ for all $p \geq 1$, we can choose an $\alpha$ to get $n^{\frac{p-1}{2}}\d{G}_n  |\Omega_n^\circ|^p = \fasterthan(1)$ as desired.
\end{proof}

We can now prove Theorem~\ref{thm:test_size}.
\begin{proof}[\bfseries{Proof of Theorem~\ref{thm:test_size}}]
Since $\Omega_0 = 0$ under $H_0$, Theorem~\ref{thm:weak_conv_omega} implies that  $n^{1/2} \Omega_n^\circ$ converges weakly as a process in $\ell^\infty(\s{A}_0)$ to $Z_0$ Thus, $n^{1/2}\|\Omega_n^\circ\|_{F_0, p} \indist \|Z_0\|_{F_0, p}$ by the continuous mapping theorem. By Lemma~\ref{lemma:Omegan_norm}, we have $n^{1/2}\|\Omega_n^\circ\|_{F_n, p} \indist  \|Z_0\|_{F_0, p}$ as well. By  Lemma~\ref{lemma:Zn_conv}, $\|Z_n\|_{F_n,p} \indist \|Z_0\|_{F_0, p}$, and since by assumption the distribution function of $\|Z_0\|_{F_0, p}$ is strictly increasing in a neighborhood of $T_{0,\alpha,p}$, the quantile function of $\|Z_0\|_{F_0, p}$ is continuous at $1-\alpha$. Therefore, $T_{n,\alpha,p}$, which is by definition the $1-\alpha$ quantile of $\|Z_n\|_{F_n,p}$, converges in probability to the $1-\alpha$ quantile of $\|Z_0\|_{F_0, p}$. Therefore, $n^{1/2}\|\Omega_n^\circ\|_{F_n, p} - T_{n,\alpha,p} + T_{0,\alpha,p} \indist  \|Z_0\|_{F_0, p}$. Hence,
\begin{align*}
P_0 \left( n^{1/2}\|\Omega_n^\circ\|_{F_n, p} > T_{n,\alpha,p}\right) &= P_0 \left( n^{1/2}\|\Omega_n^\circ\|_{F_n, p} - T_{n,\alpha,p} + T_{0,\alpha,p} > T_{0,\alpha,p}\right) \\
&\longrightarrow P_0 \left(  \|Z_0\|_{F_0, p} > T_{0,\alpha,p}\right) \leq \alpha\ . 
\end{align*}
Since by assumption the distribution function of $\|Z_0\|_{F_0,p}$ is continuous at $T_{0,\alpha,p}$, $P_0 \left(  \|Z_0\|_{F_0, p} > T_{0,\alpha,p}\right) = \alpha$, which completes the proof.

 \end{proof}
 
 \begin{proof}[\bfseries{Proof of Theorem~\ref{thm:weak_conv_local}}]
 By Theorem~\ref{thm:weak_conv_omega} and since $\Omega_0(a) = 0$ for all $a$,
 \[ \sup_{a \in \s{A}_0} \left| n^{1/2} \Omega_n(a) - \d{G}_n D_{a,0}^* \right| \inprob 0\ .\]
The distribution $P_n$ is contiguous to $P_0$ by Lemma~3.10.11 of \cite{van1996weak}. Therefore, by Theorem~3.10.5 of \cite{van1996weak},
  \[ \sup_{a \in \s{A}_0} \left| n^{1/2} \Omega_n(a) - \d{G}_n D_{a,0}^* \right| \stackrel{\mathrm{P_n}}{\longrightarrow} 0\ .\]
  Since $\{D_{a, 0}^* : a \in \s{A}_0\}$ is a $P_0$-Donsker class and $\sup_{a \in \s{A}_0} | P_0 D_{a, 0}^*|< \infty$, Theorem~3.10.12 of \cite{van1996weak} implies that $\left\{ \d{G}_nD_{a,0}^* : a \in \s{A}_0 \right\}$ converges weakly in $\ell^\infty(\s{A}_0)$ to $\{  Z_0(a) + P_0( h D_{a, 0}^*) : a \in \s{A}_0 \}$. The result follows.
 \end{proof}

 \begin{proof}[\bfseries{Proof of Theorem~\ref{thm:local_alt_power}}]
 By Lemma~\ref{lemma:Omegan_norm}, $n^{1/2}\left( \| \Omega_n^\circ\|_{F_n, p} -  \| \Omega_n^\circ\|_{F_0, p}\right) \inprob 0$. Therefore, since $P_n$ is contiguous to $P_0$ by Lemma~3.10.11 of \cite{van1996weak}, Theorem~3.10.5 of \cite{van1996weak} implies that $n^{1/2}\left( \| \Omega_n^\circ\|_{F_n, p} -  \| \Omega_n^\circ\|_{F_0, p}\right)  \stackrel{\mathrm{P_n}}{\longrightarrow}  0$. Hence, by the continuous mapping theorem and Theorem~\ref{thm:weak_conv_local}, $n^{1/2}  \| \Omega_n^\circ\|_{F_n, p}$ converges in distribution under $P_n$ to $\|\bar{Z}_{0,h}\|_{F_0,p}$. In addition, since $T_{n,\alpha,p} \inprob T_{0,\alpha,p}$ (as demonstrated in the proof of Theorem~\ref{thm:test_size}), $T_{n,\alpha,p} \stackrel{\mathrm{P_n}}{\longrightarrow}  T_{0,\alpha,p}$. Therefore, $n^{1/2}  \| \Omega_n^\circ\|_{F_n, p} - T_{n,\alpha,p}  + T_{0,\alpha,p}$ converges in distribution under $P_n$  to $\|\bar{Z}_{0,h}\|_{F_0,p}$. Thus, 
 \begin{align*}
  P_n\left( n^{1/2}  \| \Omega_n^\circ\|_{F_n, p} > T_{n,\alpha,p}\right) &= P_n\left( n^{1/2}  \| \Omega_n^\circ\|_{F_n, p} - T_{n,\alpha,p} + T_{0,\alpha,p} > T_{0,\alpha,p}\right)\\
  &\longrightarrow P\left( \|\bar{Z}_{0,h}\|_{F_0,p} > T_{0,\alpha,p}\right) \ .
 \end{align*}
 \end{proof}


\clearpage 

\section*{Additional simulation results}

Here, we present the results of the first simulation study using parametric nuisance estimators. Figure~\ref{fig:sim_size} displays the empirical type I error rate (i.e.\ the fraction of tests with $p < 0.05$) of nominal $\alpha = 0.05$ level tests using the parametric nuisance estimators across the two types of null hypotheses and three sample sizes.  The tests with correctly-specified parametric outcome regression estimators of the nuisances (first and third columns from the left) had empirical error rates within or slightly below Monte Carlo error of the nominal rate at all sample sizes and under both the strong and weak nulls. The fact that the tests with both nuisance estimators correctly specified yields correct size empirically validates the large-sample theoretical guarantee of Theorem~\ref{thm:test_size}. That the test achieved valid size when the propensity score estimator was mis-specified was not guaranteed by our theory, and we would not expect this to always be the case. The tests with $\mu_n$ based on an incorrectly specified parametric model and $g_n$ based on a correctly specified parametric model (second column of Figure~\ref{fig:sim_size}), had empirical type I error rates below the nominal rate. The tests with both $\mu_n$ and $g_n$ based on incorrectly specified parametric models (second row of Figure~\ref{fig:sim_size}), had empirical type I error rates far above the the nominal rate and converging to 1. These results align with our expectations that, in general, no guarantees with regard to type I error can be made when the nuisance estimators are inconsistent. 
 
 Figure~\ref{fig:sim_power} displays the empirical power (i.e.\ the fraction of tests with $p < 0.05$) of nominal $\alpha = 0.05$ level tests using the parametric nuisance estimators across the four types of alternative hypotheses and three sample sizes. Power increased with sample size in all cases. Given Theorem~\ref{thm:dr_cons_test}, this was expected  for the first three columns, but not necessarily for the last column, in which both nuisance estimators were inconsistent. Under alternative data-generating mechanisms, the power under inconsistent estimation of both nuisance parameters may not increase to one as the sample size increases. The power of the test was generally better further away from the null hypothesis, except when  $\mu_n$ was based on a correctly specified parametric model and $g_n$ was based on an incorrectly specified parametric model (third column).

\begin{figure}[h]
\centering
\includegraphics[width=\linewidth]{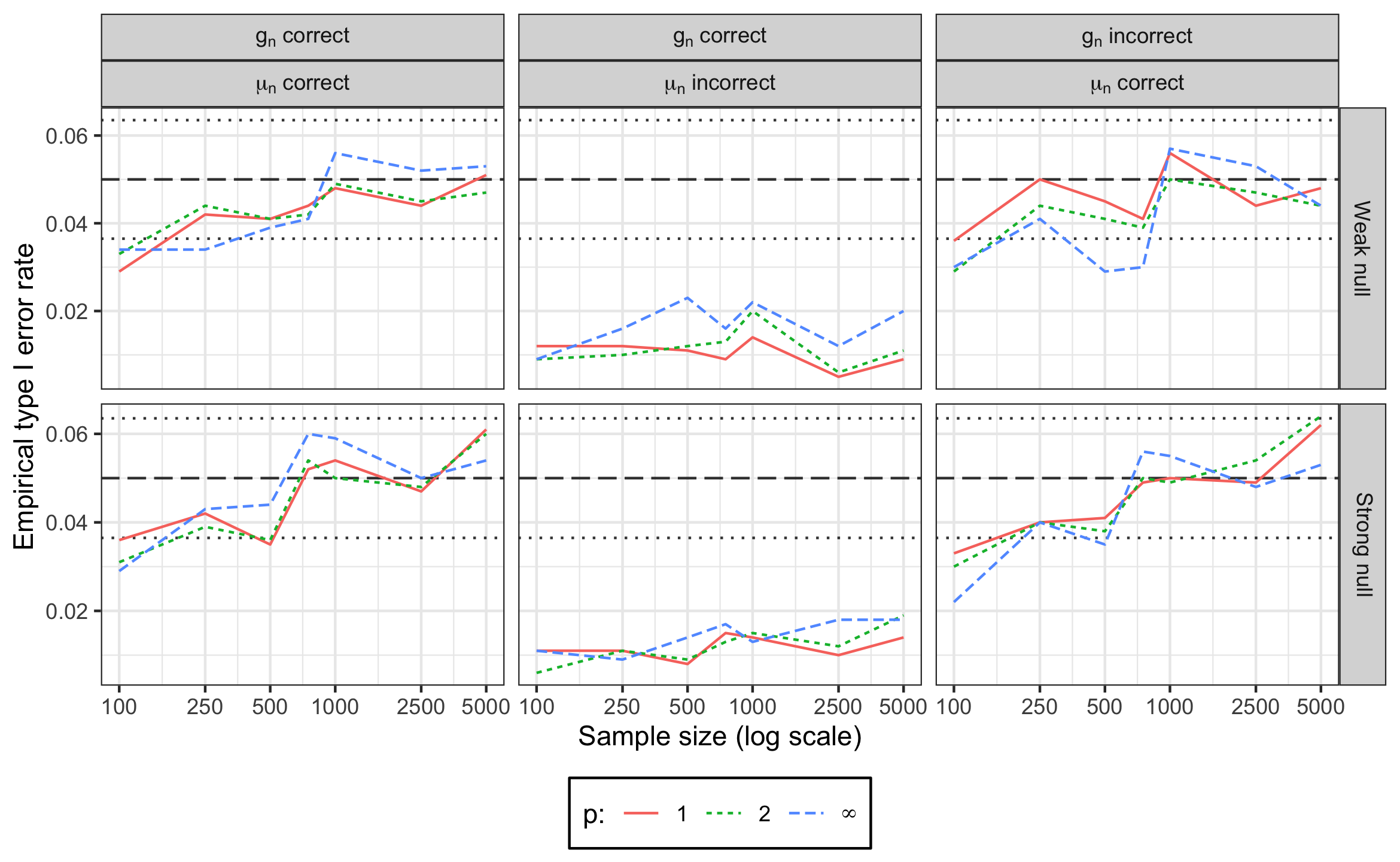}
\includegraphics[width=2.5in]{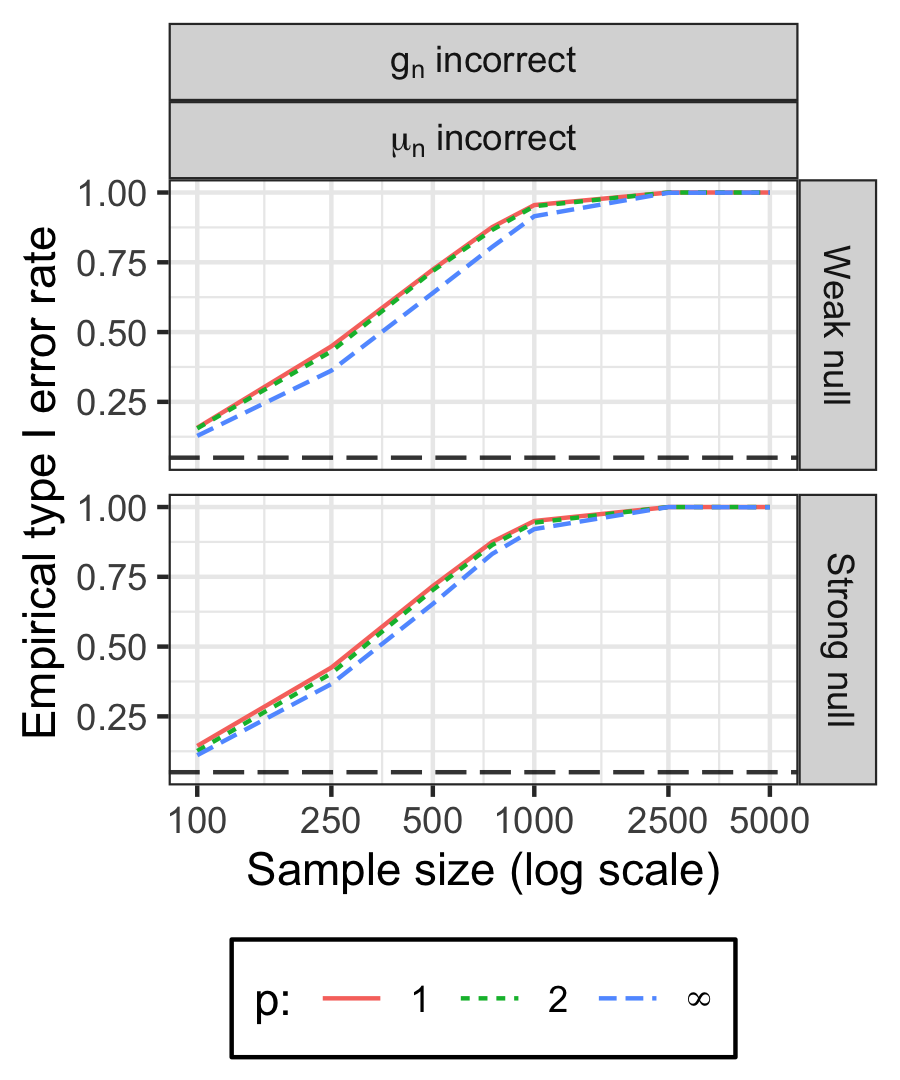}
\caption{Empirical type I error rate (i.e.\ the fraction of tests with $p < 0.05$) of nominal $\alpha = 0.05$ level tests using the parametric nuisance estimators across the two types of null hypotheses and seven sample sizes. Panels in the top row were generated under the weak null where $\mu_0$ depends on $a$, but $\theta_0$ does not. Panels in the bottom row were generated under the strong null where neither $\mu_0$ nor $\theta_0$ depends on $a$. Columns indicate the type of nuisance estimators used. Horizontal wide-dash lines indicate the nominal 0.05 test size, and horizontal dotted lines indicate sampling error bounds were the true size 0.05. In the third and fourth columns from the left, the empirical sizes are off the scale of the figure, and in particular are larger than 0.5 in all cases.}
\label{fig:sim_size}
\end{figure}

\begin{figure}[h]
\centering
\includegraphics[width=\linewidth]{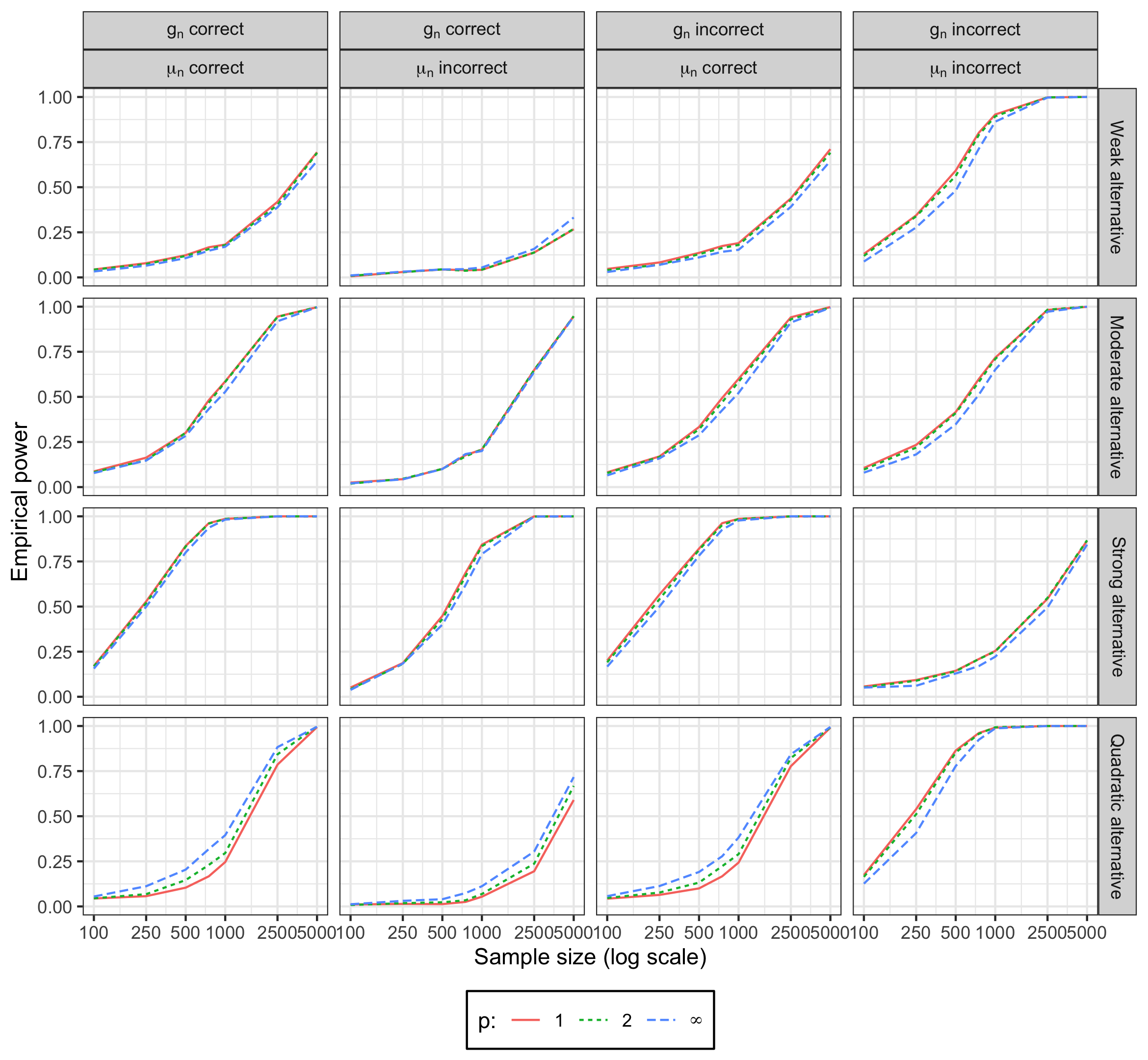}
\caption{Empirical power (i.e.\ the fraction of tests with $p < 0.05$) of nominal $\alpha = 0.05$ level tests using the parametric nuisance estimators across the four types of alternative hypotheses and three sample sizes.  Columns indicate the type of nuisance estimators used.}
\label{fig:sim_power}
\end{figure}

\clearpage 

Figure~\ref{fig:sim_power_disc} displays the empirical power of our tests in the second numerical study. We note that the number $k$ of levels of the exposure  had little impact on the power of the tests for any sample or any under alternative hypothesis. In each case, the power increased with the sample size $n$.

\begin{figure}[h]
\centering
\includegraphics[width=\linewidth]{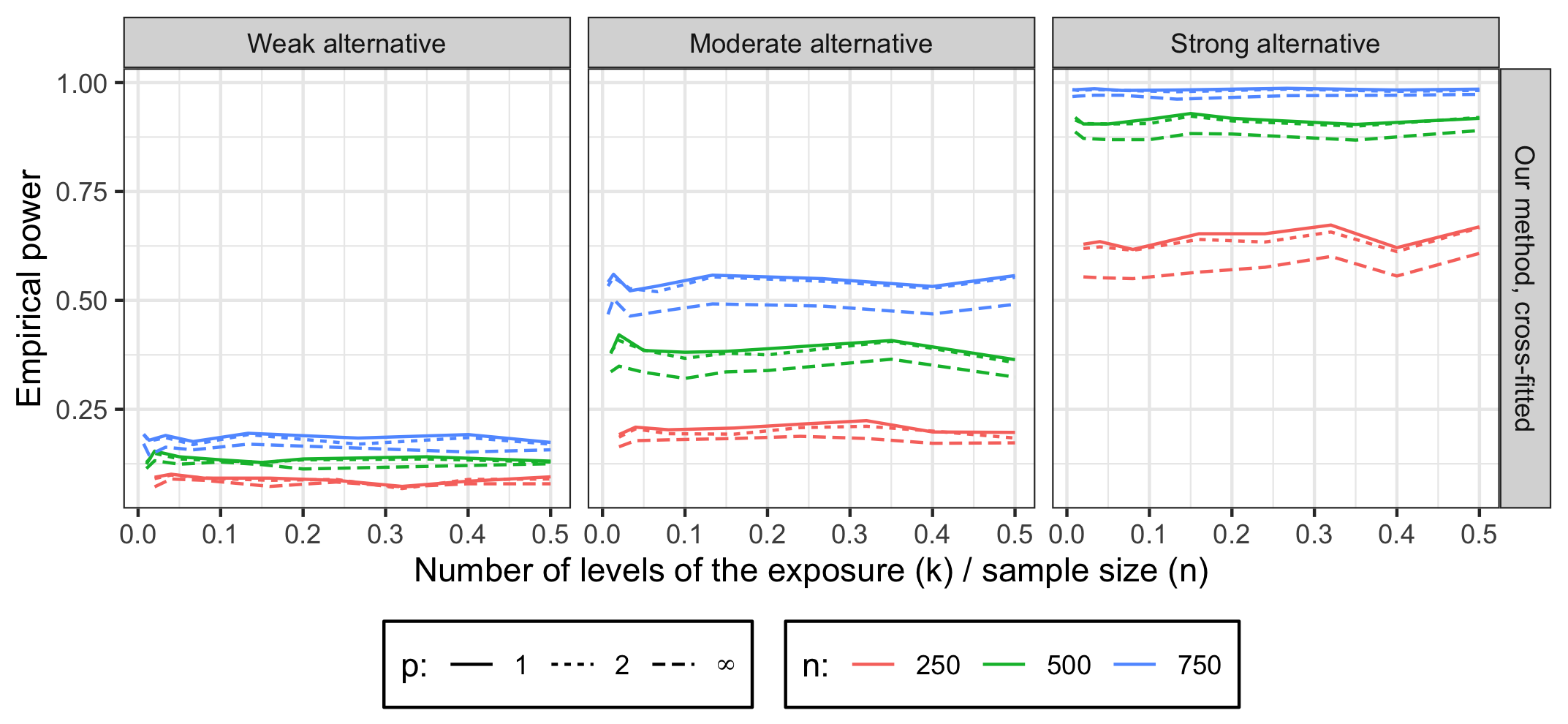}
\caption{Empirical power of nominal $\alpha = 0.05$ level tests in the second numerical study using our method with cross-fitted, correctly-specified parametric nuisance estimators. The $x$-axis is the ratio of $k$, the number of levels of the exposure, to $n$, the sample size. Columns indicate the alternative hypothesis used to generate the data.}
\label{fig:sim_power_disc}
\end{figure}

\clearpage

\section*{Additional results from analysis of the effect of BMI on immune response}\label{sec:bmi}

Figure~\ref{fig:bmi_omega} shows the estimated primitive functions $\Omega_n$ and  95\% uniform confidence bands as a function of BMI for the analysis of CD4+ responses (left) and CD8+ responses (right).

\begin{figure}[ht!]
\centering
\includegraphics[width=6in]{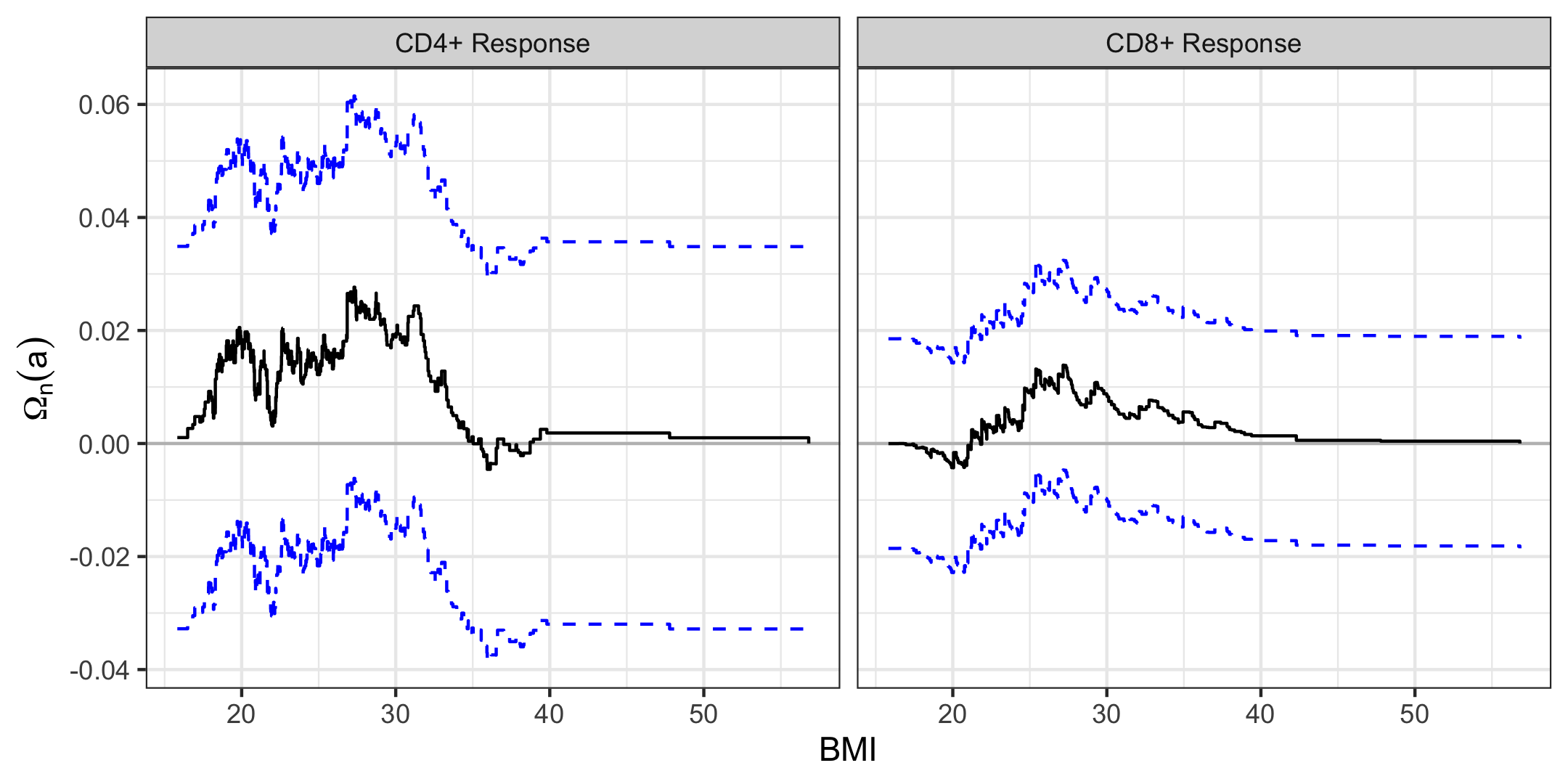}
\caption{Estimated functions $\Omega_n$ (solid line) and 95\% uniform confidence band (dashed lines) for the analysis of CD4+ responses (left) and CD8+ responses (right) presented in the article.}
\label{fig:bmi_omega}
\end{figure}

\clearpage

\bibliographystyle{apa}
\bibliography{../../biblio}

\end{document}